\documentclass[12pt,twoside]{article}
\pdfoutput=1 
\usepackage[mathscr]{eucal}
\usepackage{amsthm,amsmath,amssymb,amscd}
\usepackage{graphicx}
\usepackage{latexsym,enumerate}
\usepackage{subfig}
\usepackage{color}
\usepackage{tikz}
\usetikzlibrary{shapes.arrows,chains,positioning}
\usepackage[outline]{contour}
\newtheorem{theorem}{Theorem}
\newtheorem{lemma}{Lemma}

\theoremstyle{definition}
\newtheorem{definition}{Definition}

\def \beq{ \begin{equation} }
\def \eeq{\end{equation}}

\usepackage{color,hyperref}
\definecolor{dark-blue}{rgb}{0.15,0.15,0.4}
\definecolor{dark-red}{rgb}{0.4,0.15,0.15}
\definecolor{medium-red}{rgb}{0.6,0,0}
\definecolor{medium-blue}{rgb}{0,0,0.6}
\hypersetup{
    colorlinks=true,
    linktoc=all,
    citecolor=medium-red,
    filecolor=medium-blue,
    linkcolor=medium-blue,
    urlcolor=medium-blue,
}

\contourlength{0.065em}
\newcommand{\massConfig}[9]{
    \def\colora{#1}
    \def\colorb{#2}
    \def\colorc{#3}
    \def\colord{#4}
    \def\labela{#5}
    \def\labelb{#6}
    \def\labelc{#7}
    \def\labeld{#8}
    \def\xa{#9}
    \massConfigContinued
}
\newcommand{\massConfigContinued}[7]{
    \begin{tikzpicture}[every node/.style={draw, inner, thin, circle, inner sep=0pt, minimum size=6pt}]
        \coordinate (M1) at (\xa, #1);
        \coordinate (M2) at (#2, #3);
        \coordinate (M3) at (#4, #5);
        \coordinate (M4) at (#6, #7);
        \draw (M1)--(M2);
        \draw (M1)--(M3);
        \draw (M1)--(M4);
        \draw (M2)--(M3);
        \draw (M2)--(M4);
        \draw (M3)--(M4);

        \draw (M1) node[fill=\colora, label=\labela:\contour{white}{\tiny{$m_1$}}] {};
        \draw (M2) node[fill=\colorb, label=\labelb:\contour{white}{\tiny{$m_2$}}] {};
        \draw (M3) node[fill=\colorc, label=\labelc:\contour{white}{\tiny{$m_3$}}] {};
        \draw (M4) node[fill=\colord, label=\labeld:\contour{white}{\tiny{$m_4$}}] {};
    \end{tikzpicture}
}

\DeclareSymbolFont{AMSb}{U}{msb}{m}{n}
\DeclareMathSymbol{\bdC}{\mathbin}{AMSb}{'103}
\DeclareMathSymbol{\bdR}{\mathbin}{AMSb}{'122}
\DeclareMathSymbol{\bdN}{\mathbin}{AMSb}{'116}

\begin{document}

\title{\textbf{Bifurcations of  Central Configurations in the Four-Body Problem with some equal masses }}
\author{David Rusu \thanks{ Department of Mathematics, Wilfrid Laurier University}, 
 Manuele Santoprete\thanks{ Department of Mathematics, Wilfrid Laurier
University E-mail: msantopr@wlu.ca}}  \maketitle

\begin{abstract}
    We study the bifurcations of central configurations of the Newtonian four-body problem when some of the masses are equal. First, we continue numerically the solutions for the equal mass case, and we find values of the mass parameter at which the number of solutions changes. Then, using the Krawczyk method and some result of equivariant bifurcation theory, we rigorously prove  the existence of such bifurcations and classify them.    
\end{abstract}

\renewcommand{\thefootnote}{\alph{footnote})}


\noindent \textbf{Keywords:} celestial mechanics, n-body problem,
central configurations, bifurcations.  
\tableofcontents
\section{Introduction }
The Newtonian $n$-body problem is the study of the dynamics of $n$
point particles with masses $m_i\in{\mathbb R}^+$ and positions
$q_i\in{\mathbb R}^d$ ($i=1,\ldots,n$), moving according to
Newton's laws of motion:
\begin{equation}
m_j\ddot q_j=F _i = \sum_{i\neq j}\frac{m_im_j(q_i-q_j)}{r_{ij}^3}\quad 1\leq j\leq
n
\end{equation}
where $r_{ij}=\|q_i-q_j\|$ is the distance between $q_i$ and $q_j$. The force vector $ F _i \in \mathbb{R}  ^d $
can also be written as a partial gradient $ F _i = \nabla _i U $ where 
\[
U = \sum _{ i<j } \frac{ m _i m _j } { r _{ ij } }
\]
is the {\it Newtonian potential } function and $ \nabla _i $ denotes the vector of partial derivatives with respect to the $d$ components of $ q _i $. 

In the Newtonian $n$-body problem, the simplest possible motions are
such that the configuration is constant up to rotations and scaling,
and each body describes a Keplerian orbit. Only some special
configurations of particles are allowed in such motions. Wintner
called them {\itshape central configurations} (or c.c's, for short). A configuration
$(q_1,\ldots, q_n)$ is called a central configuration if and only if
there exists a $\lambda\in{\mathbb R}$ such that

\begin{equation}
\lambda(q_j-q_G)=\frac{ 1 } { m _j } \nabla _j U = \sum_{i\neq j} \frac{m_i(q_i-q_j)}{r_{ij}^3} \quad 1\leq
j\leq n  \label{cc}
\end{equation}
where $q_G=\sum_i m_i q_i / \sum_i m_i$ is the center of mass.
It turns out that the values of $\lambda$ are uniquely determined by the equation above, in fact 
\[\lambda = - \frac{ U } { I } \]
where 
\[
     I = \sum_i  m _i \| q _i - q _G \|^2 = \frac 1 M \sum _{i<j}  m _i m _j r _{ ij } ^2 
\]
is the moment of inertia with respect to $ q _G $, and $  M = \sum_i  m _i $.  
Equations (\ref{cc}) are invariant under rotations, dilatations and
 translations on the plane. Two central
configurations are considered equivalent if they are related by
these symmetry operations, and thus lie in the same equivalence class. 

The question of the existence and classification of central configurations is a difficult and fascinating problems that dates back to the work of 18th-century mathematicians Euler and Lagrange, and has been revived by contemporary mathematician Steven Smale \cite{smale_mathematical_1998}  with the conjecture (due to  Chazy \cite{chazy_sur_1918} and  Wintner \cite{wintner_analytical_1964}) that the number of central configurations is finite. 

An exact count of the central configurations of $n$-bodies was found by Moulton \cite{moulton_straight_1910} for the collinear $n$-body problem. Moulton showed that there are $ n!/2 $ collinear equivalence classes, that is there is one collinear relative equilibrium for each ordering of the masses.

 The number of planar central configurations  of $n$-bodies (for arbitrary $n$) is know when some of the masses are assumed sufficiently small \cite{xia_central_1991}, however, an exact count 
for an arbitrary set of positive masses is  known only  when $ n =2, 3 $. 

In the four-body problem the number of central configurations has been shown to be finite \cite{hampton_finiteness_2006}, but a complete characterization is known only for the equal masses case \cite{albouy_symetrie_1995,albouy_symmetric_1996}, when one of the masses is sufficiently small \cite{barros_set_2011,barros_bifurcations_2014}, and when there are two pairs of equal masses, with one pair sufficiently small \cite{corbera_central_2014}. There are also some partial results  if some of the masses are equal \cite{long_four-body_2002,perez-chavela_convex_2007,albouy_symmetry_2008}.

There are a number of papers investigating  the bifurcations of central configurations in the four-body problem. In \cite{simo_relative_1978} Sim\'o presented a numerical study of the bifurcations of the central configurations with arbitrary masses, and gave exact numbers of central configurations inferred by these numerical computations. In \cite{meyer_bifurcations_1988} Meyer and Schmidt studied the equilateral triangle family of central configurations and showed that  families of isosceles triangle bifurcate from the equilateral triangle family. In \cite{bernat_planar_2009} Bernat, Llibre and Perez-Chavela, studied the kite configurations of the four-body problem with three equal masses and found two bifurcation in the number of c.c.'s, one of which  is Meyer and Schmidt's bifurcation. This allowed them to obtain an exact count of the number of kite shaped c.c's.

In this paper we study the four-body problem in two special cases: the case where three of the masses are equal and the case where there are two pairs of equal masses. In both cases we first do a numerical study by varying one of the masses from the equal masses case. This allows us to determine, numerically, the values of the mass parameter for which there are bifurcations. Then we use interval arithmetic to implement the  Krawczyk method \cite{neumaier_interval_2008} and prove rigorously  the existence of the bifurcations we located numerically. In the three equal masses case we recover the bifurcations obtained in \cite{meyer_bifurcations_1988} and \cite{bernat_planar_2009} but we also find three supercritical pitchfork bifurcations for $ m = m _{ \ast \ast } \approx 0.99184227$. These are symmetry breaking bifurcations where one $\mathbb{Z}_2 $-symmetric configurations splits into three, two of which have no symmetry. In the case of two pairs of equal masses ( $ m _1 = m _2 = 1 $ and $ m _3 = m _4 = m $  with $ m \leq 1 $) we find two bifurcations: a fold and a supercritical pitchfork bifurcation.
A consequence of our analysis is that, based on our numerical results,  we are able to give an exact count of the number of c.c's in the four body problem with some equal masses. The numbers we obtain seem to be compatible with the numerical results of Sim\'o \cite{simo_relative_1978}. Unfortunately, our counts are also  based on certain numerical computations and therefore we are unable to prove the well known conjecture that states that, given four masses, there is a unique convex c.c. for each cyclic order of the masses (see Problem 10 in \cite{albouy_problems_2012}, and references therein).

Interestingly, in the four-vortex problem, a companion problem of the four-body problem, it is possible to give an exact count of the number of central configuration  if some of the vorticities are equal. In fact, in  \cite{hampton_relative_2014} we gave a complete description of the central configurations for the four-vortex problem with two pairs of equal vortices.  Unfortunately, the approach taken in \cite{hampton_relative_2014}  does not work in the Newtonian four-body problem, because the degree of the polynomial equations studied is greater and thus it is not possible to perform the same type of  Gr\"obner basis computations.


The paper is organized as follows. In Section \ref{sn:equations} we write the Dziobeck and the Albouy-Chenciner equations for central configurations. In Section \ref{sn:background} we briefly recall some important tools, namely the Krawczyk method, some bifurcation theory and some facts related to equivariant bifurcation theory.  In Section \ref{sn:3eqmasses} we study the bifurcations in the case of three equal masses. In Section \ref{sn:2eqmasses} we study the bifurcations in the case of two pairs of equal masses. 
\section{Equations of central configurations in terms of mutual distances}
\label{sn:equations}
\subsection{Dziobeck equations}
For $ n = 4 $ there are six mutual distances. A necessary and sufficient condition that six positive numbers $ r _{ ij } $, $ 1 \leq i < j \leq 4$,  are the mutual distances between four coplanar points is 
\[S = \begin{bmatrix}
        0 & 1 & 1 & 1 & 1 \\
        1 & 0 & r _{ 12 } ^2 & r _{ 13 } ^2 & r _{ 14 } ^2 \\
        1 & r _{ 12 } ^2 & 0 & r _{ 23 } ^2 & r _{ 24 } ^2 \\
        1 & r _{ 13 } ^2 & r _{ 23 } ^2 & 0 & r _{ 34 } ^2 \\
        1 & r _{ 14 } ^2 & r _{ 24 } ^2 & r _{ 34 } ^2 & 0 \\
    \end{bmatrix}.
\]
This determinant is equal to  $ 288 V ^2 $, where $V$ is the volume of the tetrahedron whose six edges are  the mutual distances $ r _{ ij } $. This formula is the three-dimensional generalization of Heron's formula for the are of a triangle.  

Using Lagrange multipliers, Dziobeck characterized the central configurations of four bodies as the critical points of 
\[
V=U + \lambda _0 (I - I _0) + \mu S
\]
viewed as a function of   eight variables  $ \lambda _0 , \mu , r _{ 12 } , r _{ 13 } , r _{ 14 } , r _{ 23 } , r _{ 24 } , r _{ 34 }   $, 
subject to the constraints $ I = I _0 $ and $ S = 0 $. Here $ \lambda _0 $ and $ \mu $ are Lagrange multipliers and $ I _0 $ is a fixed moment of inertia. Hence, the central configurations are the solution of the following  eight equations: 

\begin{align*} 
    & \frac{\partial V } { \partial \lambda _0 }   = 0, \quad \frac{ \partial V } { \partial \mu }   = 0\\
    & \frac{ \partial V } {\partial  r _{ ij } }  = 0 \quad 1 \leq i < j \leq 4.
\end{align*} 
We will denote by $ F = (F _1 , \ldots , F _8) $ the equation obtained from the equations above by clearing the denominators, with the normalization $ I _0 = 1 $, and we will refer to them as Dziobek equations. Note that this equations give only the ``strictly planar'' configurations, that is the planar configurations that are not collinear.  
\subsection{Albouy-Chenciner equations\label{ssn:AC}}

The Albouy-Chenciner equations are algebraic equations satisfied by
the mutual distances $r_{ij}$ of every central configuration \cite{albouy_probleme_1997,hampton_finiteness_2006}  
\begin{equation}
\sum_{k=1}^n m_k
[S_{ik}(r_{jk}^2-r_{ik}^2-r_{ij}^2)+S_{jk}(r_{ik}^2-r_{jk}^2-r_{ij}^2)]=0
\label{eqn:AC}
\end{equation}
for $1\leq i< j\leq n$, where $S_{ik}$ and $S_{jk}$ are given by  When $m = m _{\ast \ast} $ the $ 4 \times 4 $ submatrix obtained from the Jacobian of the Albouy-Chenciner equations by deleting the last two rows and columns has non-zero determinant. 
\begin{equation}
S_{ij}= \frac{1}{r_{ij}^3}+\lambda'\quad (i\neq j), \quad
S_{ii}=0  \label{eqn:S}
\end{equation}
where $ \lambda' =  \lambda/M $.  
Since any relative equilibria may be rescaled, we will impose the normalization  $\lambda'=-1$. This, in this case,  can be assumed without loss of generality, however, this is not true in the vortex case, \cite{hampton_relative_2014}.
 After clearing the denominators in the $ S _{ ij } $ terms, these equations form a polynomial system in the $ r _{ ij } $ variables. These new equations are also called Albouy-Chenciner (AC) equations.

In the four body case the AC equations reduce to a system of six algebraic equations in six variables (the mutual distances). 
Note that the solutions of the Albouy-Chenciner equations for the four body problem include  collinear solutions, planar solutions and one three-dimensional solution (the regular tetrahedron). It follows that the number of solution of the AC equation is equal to the number of solutions of the Dziobeck ones plus 13, since it is well known that the number of collinear solutions in the four body problem is always 12.  

\section{Theoretical Background}
\label{sn:background}
In this section we review a few theoretical facts concerning  interval arithmetic, bifurcation theory  and group actions, that we will be useful in our analysis. 
\subsection{Interval arithmetic and the Krawczyk operator}
We discuss a method to find rigorous bounds on the solution of a nonlinear smooth function $ F : \mathbb{R}^n  \to \mathbb{R}^n  $. Let $ \mathbf{x} \in \mathbb{R}^n  $, and let $ [\mathbf{x} ]_r \subset \mathbb{R}^n  $ be the interval set centered at $\mathbf{x}$ with radius $ r > 0 $. Namely,
\[
    [\mathbf{x} ] _r = \{ \mathbf{y} \in \mathbb{R}^n  : \|\mathbf{y} - \mathbf{x} \|_{ \infty } \leq r \} , 
\]
where $ \|\cdot \| _{ \infty } $ is the infinity norm. Assume the derivative of $F$ at $ \mathbf{x} $, denoted by $ DF (\mathbf{x} ) $ is nonsingular, then the {\bf Krawczyk } operator of $F$ associated with $ [\mathbf{x} ] _r $ is defined as 
\[
    K (\mathbf{x}  , [\mathbf{x} ]_r) = \mathbf{x} - DF (\mathbf{x}) ^{ - 1 } F (\mathbf{x}) + [ I - DF( \mathbf{x}) ^{ - 1 } DF ([\mathbf{x} ]_r) ] ([\mathbf{x} ] _r - \mathbf{x}).       
\]
The Krawczyk operator can be used to test the existence and uniqueness of a zero in a set $ [ \mathbf{x} ] _r $ using the following theorem 

\begin{theorem}
   Let $ F : \mathbb{R}^n  \to \mathbb{R}^n  $ be a smooth nonlinear function 
   \begin{enumerate} 
       \item If $F$ has a root $ x ^\ast \in [ \mathbf{x} ] _r $  then 
           $ x ^\ast \in [\mathbf{x} ]_r \cap K (\mathbf{x} , [\mathbf{x} ] _r) $.
       \item If $  [\mathbf{x} ]_r \cap K (\mathbf{x} , [\mathbf{x} ] _r) = \varnothing $ then $F$ has no zeroes in $ [\mathbf{x} ] _r $.
       \item If $ \varnothing \neq K (\mathbf{x} , [ \mathbf{x} ]_r )$ is a subset of the interior of $ [ \mathbf{x} ] _r $ then $F$ contains a unique zero in $ \mathbf{x} $. 
   \end{enumerate}  
\end{theorem} 
This is essentially a fixed point theorem. A proof is given in \cite{neumaier_interval_2008}. Using this theorem it is possible to implement  code to find bounds on roots of nonlinear equations. We wrote the code for Sage \cite{stein_sage_2014}, using  Sage arbitrary precision real intervals. Sage real intervals are based on the  Multiple Precision Floating-point Interval library (MPFI)  by Nathalie Revol and Fabrice Rouillier.

An interval  $[a,b] $ will often be written as a standard floating-point number with a question mark (for instance, $ 3.1416?$ ). 
The question mark indicates that the preceding digit may have an error of $ \pm 1 $. Note that in such cases usually a more precise bound is known, but it is not displayed to save space.

\subsection{Bifurcations}\label{section:bifurcations}
The saddle-node, transcritical and pitchfork bifurcations are the most important types of bifurcations that occur in system with a
system whose linearization has a one dimensional null-space.  
Let 
\[F : \mathbb{R}  ^n \times \mathbb{R}  \to \mathbb{R}^n:(\mathbf{x} , \mu )\to F (\mathbf{x} , \mu)    \]
be a smooth map, where $ \mu $ is a parameter.  We use $ D F $ to denote the Jacobian matrix, and $ F _\mu $ to denote the vector of partial derivatives of the components of $F$ with respect to $\mu$. We are interested in studying how  the number of solutions of the system $ F (\mathbf{x} , \mu) = 0 $ varies as $\mu$ varies. We have the following useful theorem, a proof of which can be found in \cite{sotomayor_generic_1973}. 
\begin{theorem}\label{thm:sotomayor} Suppose that $ F (x _0 , \mu _0) = 0 $ and that the Jacobian matrix  $A= DF (\mathbf{x} _0 , \mu _0) $ has a simple eigenvalue $ \lambda = 0 $ with eigenvector $\mathbf{v} $, and that the matrix $ A ^T $ has an eigenvector $ \mathbf{w} $ corresponding to the eigenvalue $ \lambda = 0$. Then 
\begin{enumerate}
    \item If $ \mathbf{w} ^T F _\mu (\mathbf{x} _0 , \mu _0) \neq 0 , \quad \mathbf{w} ^T [D ^2 F (\mathbf{x} _0 , \mu _0) (\mathbf{v} , \mathbf{v})   ] \neq 0 $, then the system experiences a fold bifurcation at the equilibrium point $ \mathbf{x} _0 $ as the parameter $ \mu $ passes through the bifurcation value $ \mu = \mu _0 $.   
    \item If 
        \begin{align*}
            &\mathbf{w} ^T F _\mu (\mathbf{x} _0 , \mu _0) = 0\\
            &\mathbf{w} ^T [D F _\mu (\mathbf{x} _0 , \mu _0) \mathbf{v}  ]\neq 0\\
            &\mathbf{w}^T[D ^2 F (\mathbf{x} _0 , \mu _0) (\mathbf{v} , \mathbf{v})] \neq 0    
        \end{align*} 
        then the system experiences a transcritical bifurcation at the equilibrium point $ \mathbf{x} _0 $ as the parameter $ \mu $  passes through the bifurcation value $ \mu = \mu _0 $. 
    \item If 
        \begin{align*}
            &\mathbf{w} ^T F _\mu (\mathbf{x} _0 , \mu _0) = 0\\
            &\mathbf{w} ^T [D F _\mu (\mathbf{x} _0 , \mu _0) \mathbf{v}  ]\neq 0\\
            &\mathbf{w}^T[D ^2 F (\mathbf{x} _0 , \mu _0) (\mathbf{v} , \mathbf{v})] = 0 \\
            &\mathbf{w}^T[D ^3 F (\mathbf{x} _0 , \mu _0) (\mathbf{v} , \mathbf{v})] \neq 0     
        \end{align*}   
        then the system experiences a pitchfork bifurcation at the equilibrium point $ \mathbf{x} _0 $ as the parameter $ \mu $  passes through the bifurcation value $ \mu = \mu _0 $. If $  \mathbf{w}^T[D ^3 F (\mathbf{x} _0 , \mu _0) (\mathbf{v} , \mathbf{v})] <0 $ the branches occur for $ \mu> \mu _0 $, and the bifurcation is supercritical. Otherwise, the branches occur for $ \mu < \mu _0 $ and the bifurcation is subcritical.  

\end{enumerate} 
\end{theorem} 
\subsection{Group actions and equivariant bifurcation theory}
\begin{definition}
   Let $M$ be a manifold and let $G$ be a group. A {\bf action } of a group $G$ on $M$ is a map $ \Phi :G \times M \to M $ such that: 
   \begin{enumerate}[(i)]
       \item $ \Phi (E , x) = x $, for all $ x \in M $,  where $E$ is the identity element of $G$; and
       \item $ \Phi (g , \phi (h, x)) = \Phi (gh , x) $ for all $ g, h \in G $ and $ x \in M $.      
   \end{enumerate} 
\end{definition}

For every $ g \in G $ let $ \Phi _g : M \to M : x\to \Phi (g, x) $; then (i) becomes $ \Phi _E = \operatorname{id} _M $ while  (ii) becomes $ \Phi _{ gh } = \Phi _g \circ \Phi _h $. 
In the special but important case where  $M$ is a vector space $V$  and each $ \Phi _g $ a linear transformation, the action of $G$ on $V$ is called a {\bf linear representation } of $G$ on $V$.

\begin{definition}
   Let  $ M $ and $N$ be manifolds and let $ \Phi : G \times M \to M $ , $ \Psi : G \times N \to N $ be two actions. Assume that   $ F : M \to N $ is a smooth function, then  we say that $F$ is {\bf equivariant }  with respect to these actions if  for all $ g \in G $ 
   \[
       F \circ \Phi _g = \Psi _g \circ  F.
   \]  
\end{definition}
\begin{definition}
   Let $G$ be a group acting on $M$. The {\bf isotropy subgroup } of any $ x \in M $ is 
   \[
       \Sigma  _x := \{ g \in G : \Phi _g (x)   = x \} \subset G.
   \]  
   If $ \Sigma  _x$  is nontrivial then $x$ is called an {\bf isotropic point }.
\end{definition} 
\begin{definition}
   Let $ \Sigma $ be a subgroup of $G$ where $G$ is a compact Lie group acting on a vector space $V$. The {\bf fixed point subspace } of $\Sigma$ is 
   \[
       \operatorname{Fix } (\Sigma) = \{ x \in V : \Phi _{ g } (x) = x , \forall g \in \Sigma \}.   
   \] 
\end{definition} 

We are interested in the case  where the group $G=\mathbb{Z}  _2$ and $ \{ I, R \} $ is a linear representation of $ \mathbb{Z}  _2 $ in $\mathbb{R}^n$ , where $I$ is the identity and $ R $ is an $ n \times n $ matrix satisfying 
\[
R ^2 = I. 
\]
We want to show that if $ x _0 $ is $\mathbb{Z}  _2 $-symmetric, that is $ R x _0  = x _0 $ then the symmetry can be  helpful in  determining the type of bifurcation. 

\begin{lemma}\label{lem:Z2_symmetry}
   Let $ F : \mathbb{R}^n  \times \mathbb{R}  \to \mathbb{R}^n:(\mathbf{x} , \mu) \to F (\mathbf{x}  , \mu)$ be a smooth function. Suppose that $ F $ is  $ \mathbb{Z}  _2 $-equivariant for each $ \mu $ , that is $ F(R \mathbf{x} , \mu) = R F (\mathbf{x} , \mu  )$ for every $ \mu $, and let $ \mathbf{x}  _0 $ such that $ R \mathbf{x}  _0 = \mathbf{x}  _0 $. Let $ F (\mathbf{x}  _0 , \mu .0) = 0 $ and let   $ A=D F(\mathbf{x}  _0 , \mu _0 ) $. Suppose that $ A $ has a simple eigenvalue $ \lambda = 0 $ with eigenvector $ \mathbf{v} $ such that $ R \mathbf{v} = - \mathbf{v} $, and that $ A ^T $ has an eigenvector $ \mathbf{w} $ corresponding to $ \lambda = 0 $. Then 
   \begin{align*}
       & \mathbf{w} ^T F _\mu (\mathbf{x} _0 , \mu _0) = 0 \\
       & \mathbf{w} ^T [ D ^2 F (\mathbf{x} _0 , \mu _0) (\mathbf{v} , \mathbf{v})] = 0   
   \end{align*} 
\end{lemma} 
\begin{proof}
    We prove that first expression is zero.   Differentiating   $ F(R \mathbf{x} , \mu) = R F (\mathbf{x} , \mu  )$  with respect to $ \mu $ at  the point $ (\mathbf{x} _0 , \mu _0) $,  and using the fact that $ R \mathbf{x} _0 = \mathbf{x} _0 $, yields
    \[
        F _\mu ( \mathbf{x} _0  , \mu _0 ) = R F _\mu  (\mathbf{x} _0  , \mu _0   ).
    \]
    Since the symmetry of the kernel and cokernel are the same $ R \mathbf{v} = - \mathbf{v} $ implies  $ \mathbf{w} ^T R = - \mathbf{w} ^T $. Thus, applying $ \mathbf{w} ^T $ to the left of the equation above, we obtain
    \[  
        \mathbf{w} ^T F _\mu  ( \mathbf{x} _0  , \mu _0 ) =  \mathbf{w} ^T R F _\mu  (\mathbf{x} _0  , \mu _0   )=- \mathbf{w} ^T  F _\mu  (\mathbf{x} _0  , \mu _0   ).
    \]

      We now show that the second expression is zero. Differentiating  $ F(R \mathbf{x} , \mu) = R F (\mathbf{x} , \mu  )$  with respect to $\mathbf{x}$ yields 
      \[
         DF(R \mathbf{x} , \mu )(R  \mathbf{u})= R [DF (\mathbf{x}  )(\mathbf{u})]   
      \]
     differentiating again and computing the derivative at $ (\mathbf{x} _0 , \mu) $:
     \[
         D ^2 F (R \mathbf{x _0 }, \mu _0 ) (R \mathbf{u} , R \mathbf{v}) = R [D ^2 F (\mathbf{x _0 }, \mu _0 ) (\mathbf{u} , \mathbf{v}) ]     
     \]
   If we apply $ \mathbf{w} ^T $ on the left put $ \mathbf{u} = \mathbf{v} $ and assume $ R \mathbf{x _0 } = \mathbf{x _0 } $  from the equation above we obtain
   \begin{align*} 
       \mathbf{w} ^T &  [D ^2 F (\mathbf{x} _0 , \mu _0)  (\mathbf{v} , \mathbf{v})]  = \mathbf{w} ^T [D ^2 F (\mathbf{x} _0 , \mu _0) (-\mathbf{v} , -\mathbf{v})] \\
       &  = \mathbf{w} ^T [D ^2 F (\mathbf{x} _0 , \mu _0) (R\mathbf{v} , R\mathbf{v})] = \mathbf{w} ^T R [D ^2 F(\mathbf{x} _0 , \mu _0 ) (\mathbf{v} , \mathbf{v})   ]\\
       &  =- \mathbf{w} ^T [D ^2 F (\mathbf{x} _0 , \mu _0) (\mathbf{v} , \mathbf{v})   ]
   \end{align*} 
\end{proof}

Another useful result  is the following (see \cite{kuznetsov_elements_2004}):

\begin{lemma}\label{lem:sym_alternative}
   Suppose $ T: \mathbb{R}^n  \to \mathbb{R}^n$ is a linear operator and $A$ is its matrix representation. If $T $ is $ \mathbb{Z}  _2 $-equivariant then $ AR = RA $. Suppose that the kernel of $A$ is one-dimensional, then  $ A \mathbf{v} = 0 $  implies $ R \mathbf{v} = \mathbf{v} $ or $ R \mathbf{v} = - \mathbf{v} $. 

\end{lemma}

\section{The Case of Three Equal Masses}
\label{sn:3eqmasses} 
In this section we study the bifurcations of the four body problem with three equal masses. In the first subsection we show that the equation of the central configurations are equivariant with respect to the group $ D _6 $. In the following section we give an overview of the three bifucations we found. In the last three subsections of this section we analyze each of the bifurcations in detail.
\subsection{Equivariance}
Recall that   the  dihedral group of order six, is a group with six elements  $D _6 =\{   E, g _1 , g _2 , g _3 , g _4, g _5 \}  $ with  $ g _3 = g _1 g _2 g _1$, $ g _4 = g _1 g _2 $, $ g _5 = g _2 g _1 $ and  Cayley table  

\[\begin{array}{c|cccccc} 
     \circ  &  E   &   g _1  &  g _2  &  g _3  &  g _4  &  g _5  \\ 
    \hline\\  [-2ex]
     E  &  E & g _1  & g _2 & g _3 & g _4 & g _5 \\
     g _1 & g _1 &  E & g _4 & g _5 & g _2 & g _3 \\
     g _2 & g _2 & g _5 & E & g _4 & g _3 & g _1  \\
     g _3 & g _3 & g _4 & g _5 & E & g _1 & g _2 \\
     g _4 & g _4 & g _3 & g _1 & g _2 & g _5 & E\\
     g _5 & g _5 & g _2 & g _3 & g _1 & E & g _4  \\ 
\end{array}\] 
where $E$ is the identity. This group is isomorphic to the symmetric group of degree three.  The proper subgroups of $ D _6 $ are $ \{ E \} $ (the trivial group), $ \{ E, g _1 \} $, $ \{ E, g _2 \} $ , $ \{ E , g _3 \} $, and  $\{ E, g _4, g _5 \}$ (the cyclic group of order 3).   
Consider the four body problem with three equal masses, for example let $ m _1 = m _2 = m _3  = 1 $ and $ m _4 = m $, and  consider the action $ \Phi $  of the dihedral group $ D _6 $   on $ \mathbb{R}  ^8 $ defined by 
\begin{align*}
   & \Phi _E =e: (\lambda _0 , \mu , r _{ 12 } , r _{ 13 } , r _{ 14 } , r _{ 23 } , r _{ 24 } , r _{ 34 }) \rightarrow (\lambda _0 , \mu , r _{ 12 } , r _{ 13 } , r _{ 14 } , r _{ 23 } , r _{ 24 } , r _{ 34 }) \\ 
  &  \Phi _{g _1 } = \gamma _1 : (\lambda _0 , \mu , r _{ 12 } , r _{ 13 } , r _{ 14 } , r _{ 23 } , r _{ 24 } , r _{ 34 }) \rightarrow (\lambda _0 , \mu , r _{ 13 } , r _{ 12 } , r _{ 14 } , r _{ 23 } , r _{ 34 } , r _{ 24 }) \\  
&  \Phi _{g _2 } = \gamma _2 : (\lambda _0 , \mu ,r _{ 12 } , r _{ 13 } , r _{ 14 } , r _{ 23 } , r _{ 24 } , r _{ 34 }) \rightarrow (\lambda _0 , \mu ,r _{ 23 } , r _{ 13 } , r _{ 34 } , r _{ 12 } , r _{ 24 } , r _{ 14 }) \\
  &  \Phi _{g _3 } = \gamma _3 : (\lambda _0 , \mu ,r _{ 12 } , r _{ 13 } , r _{ 14 } , r _{ 23 } , r _{ 24 } , r _{ 34 }) \rightarrow (\lambda _0 , \mu , r _{ 12 } , r _{ 23 } , r _{ 24 } , r _{ 13 } , r _{ 14 } , r _{ 34 }) \\
  &  \Phi _{g _4 } = \gamma _4 : (\lambda _0 , \mu , r _{ 12 } , r _{ 13 } , r _{ 14 } , r _{ 23 } , r _{ 24 } , r _{ 34 }) \rightarrow (\lambda _0 , \mu , r _{ 13 } , r _{ 23 } , r _{ 34 } , r _{ 12 } , r _{ 14 } , r _{ 24 }) \\
  &  \Phi _{g _5 } = \gamma _5 : (\lambda _0 , \mu , r _{ 12 } , r _{ 13 } , r _{ 14 } , r _{ 23 } , r _{ 24 } , r _{ 34 }) \rightarrow (\lambda _0 ,\mu , r _{ 23 } , r _{ 12 } , r _{ 24 } , r _{ 13 } , r _{ 34 } , r _{ 14 }) \\
\end{align*} 
Then, for each fixed value of $ m _4 $  we can think of   
the Dziobeck equations as a map $ F: \mathbb{R}  ^8 \to \mathbb{R}  ^8 $.  A computation shows that this map is  equivariant with respect to the action $ \Phi $ for each value of $ m _4 $.

Similarly one can consider the action $ \Phi $  of the dihedral group $ D _6 $   on $ \mathbb{R}  ^6 $ defined by 
\begin{align*}
   & \Phi _E =e: (r _{ 12 } , r _{ 13 } , r _{ 14 } , r _{ 23 } , r _{ 24 } , r _{ 34 }) \rightarrow (r _{ 12 } , r _{ 13 } , r _{ 14 } , r _{ 23 } , r _{ 24 } , r _{ 34 }) \\ 
  &  \Phi _{g _1 } = \gamma _1 : (r _{ 12 } , r _{ 13 } , r _{ 14 } , r _{ 23 } , r _{ 24 } , r _{ 34 }) \rightarrow (r _{ 13 } , r _{ 12 } , r _{ 14 } , r _{ 23 } , r _{ 34 } , r _{ 24 }) \\  
&  \Phi _{g _2 } = \gamma _2 : (r _{ 12 } , r _{ 13 } , r _{ 14 } , r _{ 23 } , r _{ 24 } , r _{ 34 }) \rightarrow (r _{ 23 } , r _{ 13 } , r _{ 34 } , r _{ 12 } , r _{ 24 } , r _{ 14 }) \\
  &  \Phi _{g _3 } = \gamma _3 : (r _{ 12 } , r _{ 13 } , r _{ 14 } , r _{ 23 } , r _{ 24 } , r _{ 34 }) \rightarrow (r _{ 12 } , r _{ 23 } , r _{ 24 } , r _{ 13 } , r _{ 14 } , r _{ 34 }) \\
  &  \Phi _{g _4 } = \gamma _4 : (r _{ 12 } , r _{ 13 } , r _{ 14 } , r _{ 23 } , r _{ 24 } , r _{ 34 }) \rightarrow (r _{ 13 } , r _{ 23 } , r _{ 34 } , r _{ 12 } , r _{ 14 } , r _{ 24 }) \\
  &  \Phi _{g _5 } = \gamma _5 : (r _{ 12 } , r _{ 13 } , r _{ 14 } , r _{ 23 } , r _{ 24 } , r _{ 34 }) \rightarrow (r _{ 23 } , r _{ 12 } , r _{ 24 } , r _{ 13 } , r _{ 34 } , r _{ 14 }) \\
\end{align*} 
Then, for each fixed value of $ m _4 $  we can think of   
the Albouy-Chenciner equations as a map $ f: \mathbb{R}  ^6 \to \mathbb{R}  ^6 $.  A computation shows that this map is  equivariant with respect to the action $ \Phi $ for each value of $ m _4 $.    
\subsection{The global picture}\label{ssn:global_picture_1}
In this subsection we give an overview of the bifurcations we found in the case three of the masses are equal. Many of the remarks are based on numerical computations. 
The central configurations of four bodies with four equal masses are well understood since the work of Alain Albouy \cite{albouy_symetrie_1995,albouy_symmetric_1996}.  Hence, one can compute numerical approximations of the central configurations and then use continuation methods to find the central configurations for as the mass parameter varies.   We computed the solutions  of the  central configuration equations using homotopy continuation methods  for the equal mass case  using HOM4PS2  \cite{lee_hom4ps-2.0:_2008} for the Dziobeck equations, and  HOM4body (an offshoot of HOM4PS2 ) for the Albouy-Chenciner equations . We then numerically studied the solutions of the  equations as we varied the parameter  $ m=m_4 $.  We found that the Jacobian determinant of the  equations vanishes along certain solutions at $ m =m _\ast = (81 + 64 \sqrt{3}) / 249  \approx 0.77048695$, $ m =m _{ \ast \ast }  \approx 0.99184227$ and at $ m = m ^\ast \approx 1.00266054$. Each of these values of $m$  actually correspond to  a bifurcation. 

At $ m = m ^\ast $ there are three fold bifurcations (or turning points). In this case, as $m $ is increased through $ m ^\ast $ six solutions coalesce to three. These solutions are illustrated in Figure  \ref{fig:third_bifurcation}.

At $ m = m _{ \ast \ast } $ there  are three supercritical pitchfork bifurcation, so that, when decreasing $m$,  nine solutions coalesce to three. Some of these solutions are illustrated in Figure \ref{fig:first_bifurcation}, where we show one of the pitchfork bifurcations.  The other two cases are similar, except that the solutions to be considered in the two remaining cases  have $ m _1 $ in the convex hull formed by the other three masses, in one case, and $ m _2 $ in the other.


At $ m=m _{ \ast} =(81 + 64 \sqrt{3}) / 249 \approx 0.77048695$ four solutions coalesce into one, and then, as $ m  $ decreases the one solution branches into four solutions again (see Figure \ref{fig:second_bifurcation} ). This value of $ m$  can easily be found analytically by studying the equilateral triangle family $ r _{ 12 } = r _{ 13 } = r _{ 23 } = 1  $, $ r _{ 14 } = r _{ 24 } = r _{ 34 } = \frac{ \sqrt{3}}{3} $. With the aid of a computer algebra system one can show that the  value of the Jacobian determinant of the Dziobeck equations (with the normalizing condition $ I = 1 $ ) along the equilateral family is 

\[
- \frac{ 64(60 \sqrt{3} - 133)(- 249 m + 64 \sqrt{3} + 81)^2m^2(m + 3)^5} { 20667 } 
\]
which is non-zero for all positive values of $m$ except for $ m  = m _{\ast\ast}  =(81 + 64 \sqrt{3}) / 249 \approx 0.77048695 $.
  The  value of $ m _{\ast\ast} $ was originally found  analytically by Palmore  \cite{palmore_relative_1973} and an analytical study of the bifurcations at this point was done in \cite{meyer_bifurcations_1988}. 

 The number of solutions to the Dziobeck and Albouy-Chenciner equations together with the number of geometrically distinct planar central configurations implied by our numerical computations is summarized in the following table

 \begin{center}
    \begin{tabular}{| l | l | l | l |}
    \hline
     Value of $ m _4 $ &  \# of solns & \# of solns & \# of geometrically  \\ 
     &  of Dziobek eqns&  of AC eqns & different c.c's  \\ \hline\hline
    $ (0,m _\ast )$ & 13 & 26 & 25\\ \hline
    $ m_\ast$  & 10 & 23 & 22\\ \hline
    $(m_\ast, m _{ \ast \ast }]  $ & 13 & 26 & 25\\ \hline
    $ (m_{\ast\ast}, m ^\ast)  $  & 19 & 32 & 31\\ \hline 
    $ m ^\ast $  & 16 & 29 & 28\\ \hline
    $ (m ^\ast, \infty)  $  & 13 & 26 & 25\\ \hline 
    \end{tabular}
\end{center}

\subsection{Bifurcation at \texorpdfstring{$ m=m ^\ast \approx  1.00266054$}{}}

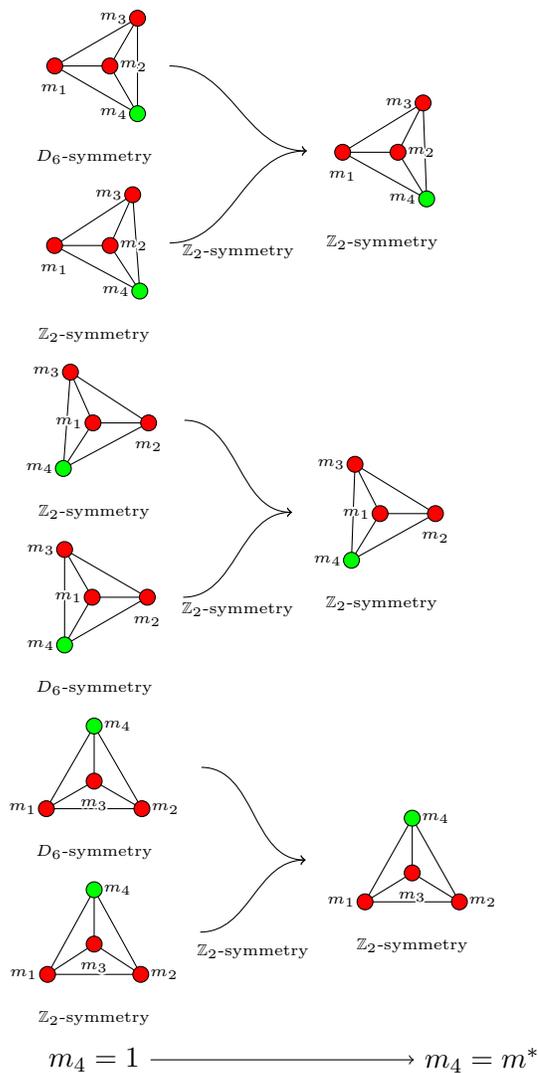
\begin{figure}  
\begin{center}
\begin{tikzpicture}[remember picture,inner/.style={anchor=center},outer/.style={},scale=1, every node/.style={scale=1}]
    \node[outer,label= below:{\tiny $ D _6 $-symmetry }] (1) {
        \massConfig{red}{red}{red}{green}
                   {below}{right}{left}{left}
                   {0}{0}
                   {0.733313000286325}{0}
                   {1.09996950042949}{0.635067687173342}
                   {1.09996950042949}{-0.635067687173343}
    };
    \node[outer,below=0.4 of 1,label=below:{\tiny $\mathbb{Z}_2$-symmetry}] (2) {
        \massConfig{red}{red}{red}{green}
                   {below}{right}{left}{left}
                   {0}{0}
                   {0.739038147757465}{0.0}
                   {1.03871007800198}{0.675554674370853}
                   {1.13317732277927}{-0.606015099535512}
    };

    \node[outer,below=0.4 of 2,label=below:{\tiny $\mathbb{Z}_2$-symmetry}] (3) {
        \massConfig{red}{red}{red}{green}
                   {left}{below}{left}{left}
                   {0}{0}
                   {0.739038147757465}{0}
                   {-0.299671930244515}{0.675554674370855}
                   {-0.394139175021803}{-0.606015099535511}
    };
      \node[outer,below=0.4 of 3,label=below:{\tiny $D _6 $-symmetry }] (4) {
        \massConfig{red}{red}{red}{green}
                   {left}{below}{left}{left}
                   {0}{0}
                   {0.733313000286325}{0}
                   {-0.366656500143161}{0.635067687173343}
                   {-0.366656500143161}{-0.635067687173342}
    };

    \node[outer,below=0.4 of 4 ,label=below:{\tiny $D_6$-symmetry}] (5) {
        \massConfig{red}{red}{red}{green}
                   {left}{right}{below}{right}
                   {0}{0}
                   {1.27013537434668}{0}
                   {0.635067687173341}{0.366656500143163}
                   {0.635067687173345}{1.09996950042949}
    };

      \node[outer,below=0.4 of 5,label=below:{\tiny $\mathbb{Z}_2$-symmetry}] (6) {
        \massConfig{red}{red}{red}{green}
                   {left}{right}{below}{right}
                   {0}{0}
                   {1.23906930565130}{0}
                   {0.619534652825650}{0.402932001445636}
                   {0.619534652825651}{1.12584277751133}
    };
    
      \node[outer, above right =2.3of 3,label=below:{\tiny $\mathbb{Z}_2$-symmetry}] (7) {
        \massConfig{red}{red}{red}{green}
                   {below}{right}{left}{left}
                   {0}{0}
                   {0.735928393771000}{0}
                   {1.06938165100555}{0.656048417423534}
                   {1.11666101052609}{-0.620607295569213}
    };

      \node[outer,below right=2.3 of 2,label=below:{\tiny $\mathbb{Z}_2$-symmetry}] (8) {
        \massConfig{red}{red}{red}{green}
                   {left}{below}{left}{left}
                   {0}{0}
                   {0.735928863812000}{0}
                   {-0.333447752021637}{0.656051742837048}
                   {-0.380734870169108}{-0.620604912453955}
    };

      \node[outer,below right =2.3 of 4,label=below:{\tiny $\mathbb{Z}_2$-symmetry}] (9) {
        \massConfig{red}{red}{red}{green}
                   {left}{right}{below}{right}
                   {0}{0}
                   {1.25458225983000}{0}
                   {0.627291129916760}{0.384833000836338}
                   {0.627291129904817}{1.11292006160845}
    };

    \path (1) edge [->,out=0,in=180] node {} (7);
    \path (2) edge [->,out=0,in=180] node[label={[label distance=14]below:\tiny{$\mathbb{Z}_2 $}-symmetry}] {} (7);
    \path (3) edge [->,out=0,in=180] node {} (8);
    \path (4) edge [->,out=0,in=180] node[label={[label distance=14]below:\tiny{$\mathbb{Z}_2 $}-symmetry}] {} (8);
    \path (5) edge [->,out=0,in=180] node {} (9);
    \path (6) edge [->,out=0,in=180] node[label={[label distance=14]below:\tiny{$\mathbb{Z}_2 $}-symmetry}] {} (9);
    \node[outer,below =0.5 of 6] (10) {{\small $ m_4=1 $} };
    \node[outer,right=3.5 of 10] (11) {{\small$ m_4=m ^\ast  $ }};
    \path (10) edge [->,out=0,in=180] node[above] {} (11);%
\end{tikzpicture}
\end{center}
\caption{On the left we show three pairs of solutions for $m _1 = m _2 = m _3 = m _4 = 1$. These solutions are continued, by increasing the parameter $ m _4=m $. Then, at $ m =m ^\ast   \approx 1.00266054...$, each pair of solution  coalesce into one solution with a $  \mathbb{Z}  _2 $ symmetry. This solution cannot be continued further, since we encounter a fold bifurcation.  \label{fig:third_bifurcation}} 
\end{figure}

We now use interval arithmetic to analyze one of the folds bifurcations at $ m = m ^\ast $ (the one with $ m _3 $ in the convex hull formed by the other masses). This approach will allow us to prove that the bifurcation exists and it is a fold. 
Let $ \tilde F = [(F _1 , \ldots , F _8, \det(DF)] $  be the vector having as components the Dziobeck equations and the determinant of the Jacobian matrix of $F$. Then we can use  the Krawczyk operator to prove the existence of a (unique) solution $(\mathbf{x}  ^\ast , m ^\ast )$ to the equation $  \tilde F (\mathbf{x} , m ) = 0 $ in a  small box. Let  $ [\mathbf{x} ^\ast]\times [m ^\ast ]$ be the box containing the solution  $(\mathbf{x}  ^\ast , m ^\ast )$. Using as initial guess a value obtain using numerical computations we obtain that
\[[\mathbf{x} ^\ast]=\begin{bmatrix}
    4.10486749931246396567394557?\\
    0.7904883951465367?\\
    0.98742601345653?\\
    0.57921860462471?\\
    1.00549177029900?\\
    0.57921860462471?\\
    1.00549177029900?\\
    0.57304559793134?\\
\end{bmatrix}\]
and $ [m ^\ast ]= 1.00266054757261000068580350?$.
Suppose $ A = DF ([\mathbf{x}  ^\ast ],[ m ^\ast]) $. Computing the echelon form of $A$ using Gauss elimination it is possible to show rigourosly that the null-space  of $A$ is one dimensional, since we know that at least one eigenvalue must be zero, but seven of the eight rows of the echelon form are clearly non-zero.   From the echelon form of $A$ we  find that the eigenvectors of $ A $ and $ A ^T $ corresponding to the zero eigenvalue are  
\[\mathbf{v} = 
\left[\begin{array}{r}
0.? \times 10^{-9} \\
-0.179026448? \\
2.989514215? \\
-0.5496816801? \\
-1.4331568126? \\
-0.5496816801? \\
-1.4331568126? \\
1
\end{array}\right], \mbox{ and}\quad \mathbf{w} =
\left[\begin{array}{r}
0.? \times 10^{-9} \\
-0.235312131? \\
1.0068617795? \\
-0.5380276784? \\
-0.46549501352? \\
-0.5380276784? \\
-0.46549501352? \\
1
\end{array}\right],
\]
respectively.
Moreover we have that 
\begin{align*} 
    & \mathbf{w} ^T F _m ([\mathbf{x}  ^\ast ],[m ^\ast ]) = -6.501134640? \\
    & \mathbf{w} ^T [D^2F ([\mathbf{x}  ^\ast ],[m ^\ast ])(\mathbf{v} ,\mathbf{v} ) ]= -2066.64414?
\end{align*}  
and thus, since the interval obtained do not contain zero,  by Theorem \ref{thm:sotomayor}, the bifurcation occurring at   $(\mathbf{x}  ^\ast , m ^\ast )$ is a fold bifurcation.

This bifurcation can also be studied by imposing the symmetry on the equations. This approach was taken in \cite{bernat_planar_2009}. Note the bifurcation value we obtain  differs slightly from the one obtained in \cite{bernat_planar_2009}. We are confident that our value for $ m ^\ast $ is the correct one since we verified it using several different methods (including using the equations used in  \cite{bernat_planar_2009}).

\subsection{Bifurcation at \texorpdfstring{$ m = m _{ \ast \ast }\approx   0.99184227$}{} \label{sec:bifurcation_at_0.99184227}}

\begin{figure}
\begin{center}
\begin{tikzpicture}[remember picture,inner/.style={anchor=center},outer/.style={},scale=1, every node/.style={scale=1}]
    \node[outer,label=below:{\tiny $\mathbb{Z}_2$-symmetry}] (A1) {
        \massConfig{red}{red}{red}{green}
                   {left}{right}{right}{right}
                   {0}{0}
                   {1.28504675000000}{0}
                   {0.651697123470713}{0.348522949625159}
                   {0.597368443982243}{1.08556146630125}
    };
    
    \node[outer,below=of A1,label=below:{\tiny $D _6$-symmetry}] (A2) {
        \massConfig{red}{red}{red}{green}
                   {left}{right}{right}{right}
                   {0}{0}
                   {1.27013537000000}{0.0}
                   {0.635067685000000}{0.366656503334852}
                   {0.635067685000000}{1.09996949666515}
    };

    \node[outer,below=of A2,label=below:{\tiny $\mathbb{Z}_2$-symmetry}] (A3) {
        \massConfig{red}{red}{red}{green}
                   {left}{right}{right}{right}
                   {0}{0}
                   {1.28504675000000}{0}
                   {0.633349626529287}{0.348522949625159}
                   {0.687678306017757}{1.08556146630125}
    };

    \node[outer,right =0.5 of A1, label=below:{{\tiny no symmetry}}] (A1bis) {
        \massConfig{red}{red}{red}{green}
                   {left}{right}{right}{right}
                   {0}{0}
                   {1.28097976000000}{0}
                   {0.648499054358301}{0.348229018798578}
                   {0.614609515763249}{1.08628051547614}
    };
    \node[outer,right=0.5 of A2,label=below:{\tiny $\mathbb{Z}_2 \cong  \{E, g _3  \}$}] (A2bis) {
        \massConfig{red}{red}{red}{green}
                   {left}{right}{right}{right}
                   {0}{0}
                   {1.28097976000000}{0.0}
                   {0.640489880000000}{0.353912214766993}
                   {0.640489880000000}{1.09090737569573}
    };

    \node[outer,right=0.5 of A3,label=below:{\tiny no symmetry}] (A3bis) {
        \massConfig{red}{red}{red}{green}
                   {left}{right}{right}{right}
                   {0}{0}
                   {1.28559920000000}{0}
                   {0.637100145641699}{0.348229018798577}
                   {0.670989684236751}{1.08628051547614}
    };

(0.000000000000000, 0.000000000000000, 1.28559920000000,
0.000000000000000, 0.637100145641699, 0.348229018798577,
0.670989684236751, 1.08628051547614)
     
    \node[outer,right =1 of A2bis,label=below:{\tiny $\mathbb{Z}_2 \cong  \{E, g _3  \}$}] (A4) {
        \massConfig{red}{red}{red}{green}
                   {left}{right}{right}{right}
                   {0}{0}
                   {1.22925116496338}{0}
                   {0.642974811453094}{0.348041582901409}
                   {0.642974816260335}{1.08673830646806}
    };
    \node[outer,right =0.5 of A4,label=below:{\tiny $\mathbb{Z}_2 \cong  \{E, g _3  \}$}] (A5) {
        \massConfig{red}{red}{red}{green}
                   {left}{right}{below}{right}
                   {0}{0}
                   {1.49354191000000}{0}
                   {0.746770955000000}{0.000950575205582646}
                   {0.746770955000000}{0.851635371584066}
    };
    \path (A1) edge [->,out=0,in=180] node[above] {} (A1bis);
    \path (A2) edge [->,out=0,in=180] node[above] {} (A2bis);
    \path (A3) edge [->,out=0,in=180] node[below] {} (A3bis);

    \path (A1bis) edge [->,out=0,in=180] node[above] {} (A4);
    \path (A2bis) edge [->,out=0,in=180] node[above] {} (A4);
    \path (A3bis) edge [->,out=0,in=180] node[below] {} (A4);
    \path (A4) edge [->,out=0,in=180] node[below] {} (A5);

    \node[outer,below =  of A3] (1) {{\small $ m_4=1 $} };
    \node[outer,right=2 of 1] (2) {{\small$ m_4=0.995  $ }};
    \node[outer,right=2 of 2] (3) {{\small $ m_4=m ^\ast  $} };
    \node[outer,right=1.5 of 3] (4) {{\small $ m_4=0.005 $} };
    \path (1) edge [->,out=0,in=180] node[above] {} (2);%
    \path (2) edge [->,out=0,in=180] node[above] {} (3);%
    \path (3) edge [->,out=0,in=180] node[above] {} (4);%

\end{tikzpicture}
\end{center}
\caption{On the left we show three solutions for $m _1 = m _2 = m _3 = m _4 = 1$. These solutions are continued, by varying the parameter $ m =m_4 $. As soon as $ m  <1 $ the solutions loose symmetry. Then, at $ m =m _\ast   \approx 0.99184227\ldots$,   they coalesce into one solution with a $  \mathbb{Z}  _2 $ symmetry. This solution can be continued further.  On the right we show the corresponding solutions for $ m _1 = m _2 = m _3 = 1 $ and $ m  = 0.005 $.\label{fig:first_bifurcation}  } 
\end{figure}
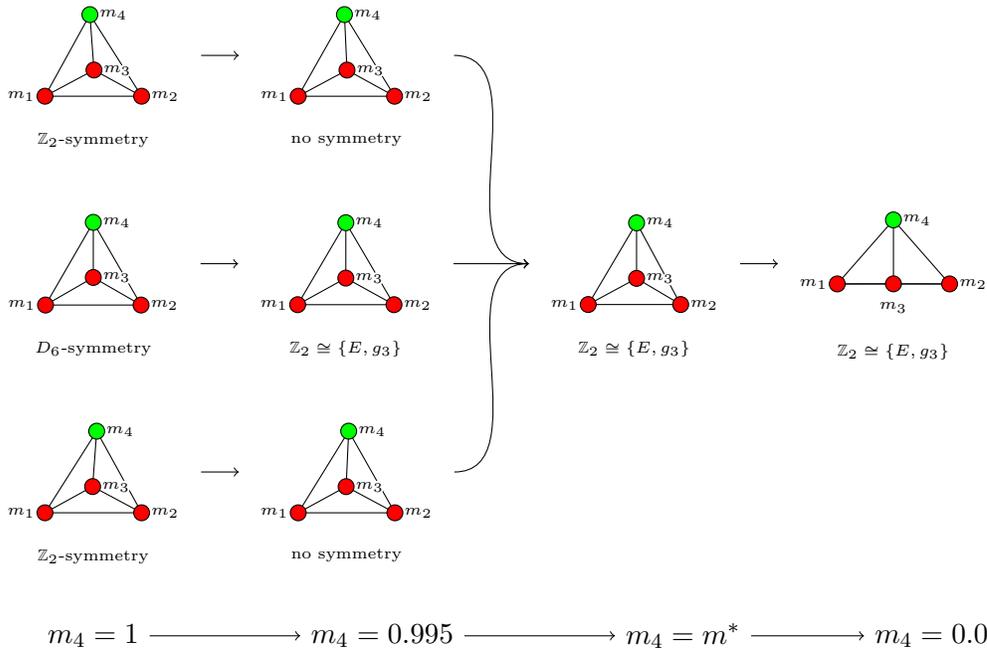

We now use interval arithmetic to analyze one of the pitchfork bifurcations  at $ m = m _{ \ast \ast} $ (the one in which $ m _3 $ is in the convex hull formed by the other masses).
Let $ \tilde F = [(F _1 , \ldots , F _8, \det(DF)] $  be the vector having as components the Dziobeck equations and the determinant of the Jacobian matrix of $F$. Then we can use  the Krawczyk operator to prove the existence of a (unique) solution $(\mathbf{x}  _{\ast\ast} , m _{\ast\ast} )$ to the equation $  \tilde F (\mathbf{x} , m ) = 0 $ in a  small box. Let  $ [\mathbf{x} _{\ast\ast}]\times [m _{\ast\ast} ]$ be the box containing the solution  $(\mathbf{x}  _{\ast\ast} , m _{\ast\ast} )$. Using as initial guess a value obtain using numerical computations we obtain that
\[[\mathbf{x} _{\ast\ast}]=\begin{bmatrix}
    4.07733304636361696432719?\\
    0.777155400247894593452215?\\
    1.013474951606110121651278?\\
    0.57621299527180?\\
    0.995153301920946?\\
    0.57621299527180?\\
    0.995153301920946?\\
    0.582177257875351248071238?\\
\end{bmatrix}\]
and $ [m _{\ast\ast} ]= 0.99184227439094091554349?$.
Suppose $ A = DF ([\mathbf{x}  _{\ast\ast} ],[ m _{\ast\ast}]) $. Computing the echelon form of $A$ using Gauss elimination it is possible to show rigorously that the null-space  of $A$ is one dimensional, since we know that at least one eigenvalue must be zero, but seven of the eight rows of the echelon form are clearly non-zero.   From the echelon form of $A$ we  find that the eigenvectors of $ A $ and $ A ^T $ corresponding to the zero eigenvalue are  
\[\mathbf{v} = \left[\begin{array}{r}
0.? \times 10^{-10} \\
0.? \times 10^{-11} \\
0.? \times 10^{-11} \\
0.34810374597? \\
-1.000000000000? \\
-0.348103745971? \\
1 \\
0
\end{array}\right]
, \mbox{ and}\quad \mathbf{w} =
\left[\begin{array}{r}
0.? \times 10^{-10} \\
0.? \times 10^{-11} \\
0.? \times 10^{-11} \\
1.03829933568? \\
-1.000000000000? \\
-1.038299335680? \\
1 \\
0
\end{array}\right],
\]
respectively.
Moreover we have that 
\begin{align*} 
    & \mathbf{w} ^T F _m ([\mathbf{x}  _{\ast \ast}],[m _{\ast \ast}]) = 0.? \times 10^{-10} \\
    & \mathbf{w} ^T [DF_m ([\mathbf{x}  _{\ast \ast}],[m _{\ast\ast} ])\mathbf{v}  ]= 34.944523147?\\
    & \mathbf{w} ^T [D^2F ([\mathbf{x}  _{\ast \ast}],[m _{\ast\ast} ])(\mathbf{v} ,\mathbf{v} ) ]= 0.?\times 10^{-7}\\
    & \mathbf{w} ^T [D^3F ([\mathbf{x}  _{\ast \ast}],[m _{\ast\ast} ])(\mathbf{v} ,\mathbf{v},\mathbf{v} ) ]= -2636.629585?\\
\end{align*}  
and thus,   by Theorem \ref{thm:sotomayor}, this suggests that the  bifurcation occurring at   $(\mathbf{x}  _{\ast\ast} , m _{\ast\ast} )$ is a pitchfork bifurcation. Since the last expression above is negative, again by Theorem \ref{thm:sotomayor}, the branches occur for $ m> m _{ \ast\ast} $ and the bifurcation is supercritical. To prove  rigorously that the bifurcation we found is  indeed a pitchfork bifurcation we use Lemma \ref{lem:Z2_symmetry}. In this case $ R=\Phi _{ g _3 }$. The fact that $ R  \mathbf{x}  ^\ast= R \mathbf{x}  ^\ast $ follow from the  symmetry of the solution (the symmetry of the solution can be shown rigorously by applying the Krawczyk operator to $ \tilde F $ with the constraints imposed by the symmetry). Also, from Lemma \ref{lem:sym_alternative} we have either $ R \mathbf{v} =  \mathbf{v} $ or  $ R \mathbf{v} = - \mathbf{v} $. Inspecting the interval expression we obtained for $\mathbf{v}$ it is clear that the first alternative cannot hold, hence $ R \mathbf{v} = - \mathbf{v} $. Thus, the hypothesis of \ref{lem:Z2_symmetry} are verified and  the bifurcation is a pitchfork.

\subsection{Bifurcation at \texorpdfstring{$ m = m _{ \ast }=(81 + 64 \sqrt{3}) / 249 \approx   0.77048695$}{}  }
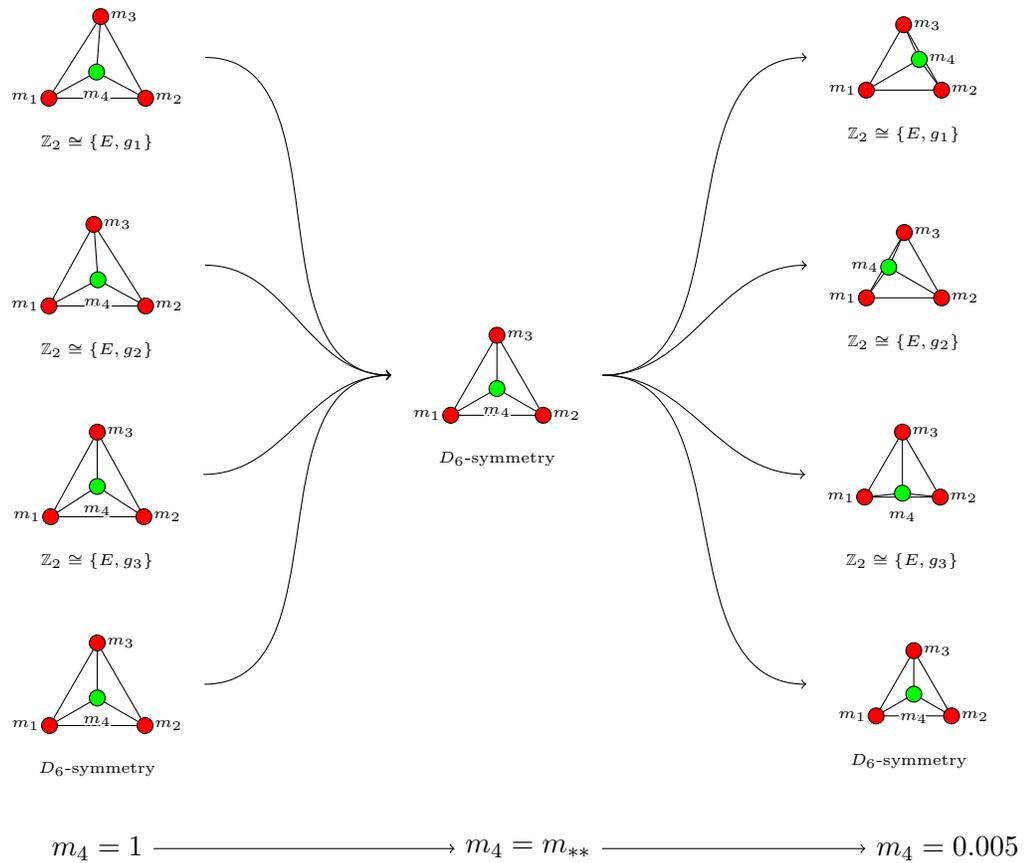
\begin{figure}
\begin{center}
\begin{tikzpicture}[remember picture,inner/.style={anchor=center},outer/.style={},scale=1, every node/.style={scale=1}]
     \node[outer,label=below:{\tiny $\mathbb{Z}_2\cong \{E, g_1 \} $}] (1) {
        \massConfig{red}{red}{red}{green}
                   {left}{right}{right}{below}
                   {0}{0}
                   {1.28504675000000}{0}
                   {0.687678306017757}{1.08556146630125}
                   {0.633349626529287}{0.348522949625159}
     };
     \node[outer,below=of 1,label=below:{\tiny $\mathbb{Z}_2 \cong \{E, g _2 \} $}] (2) {
       \massConfig{red}{red}{red}{green}
                  {left}{right}{right}{below}
                  {0}{0}
                  {1.28504675000000}{0.0}
                  {0.597368443982243}{1.08556146630125}
                  {0.651697123470713}{0.348522949625159}
    };
    
    \node[outer,below=of 2,label=below:{\tiny $\mathbb{Z}_2\cong \{E, g _3 \} $}] (3) {
        \massConfig{red}{red}{red}{green}
                   {left}{right}{right}{below}
                   {0}{0}
                   {1.23906931000000}{0}
                   {0.619534655000000}{1.12584277807321}
                   {0.619534655000000}{0.402932002215576}
    };
    \node[outer,below=of 3, label=below:{\tiny $D_6$-symmetry}] (4) {
        \massConfig{red}{red}{red}{green}
                   {left}{right}{right}{below}
                   {0}{0}
                   {1.27013537000000}{0}
                   {0.635067685000000}{1.09996949666515}
                   {0.635067685000000}{0.366656503334852}
    };
    \node[outer,below right =3.5 of 1,label=below:{\tiny $D_6$-symmetry}] (5) {
        \massConfig{red}{red}{red}{green}
                   {left}{right}{right}{below}
                   {0}{0}
                   {1.22925116496338}{0}
                   {0.614625582355866}{1.06456273654612}
                   {0.614625582513117}{0.354854245548158}
    };

    \path (1) edge [->,out=0,in=180] node[above] {} (5);
    \path (2) edge [->,out=0,in=180] node[above] {} (5);
    \path (3) edge [->,out=0,in=180] node[below] {} (5);
    \path (4) edge [->,out=0,in=180] node[below] {} (5);

    \node[outer,right=8 of 1,label=below:{\tiny $\mathbb{Z}_2\cong \{E, g _1 \} $}] (6) {
        \massConfig{red}{red}{red}{green}
                   {left}{right}{right}{right}
                   {0}{0}
                   {1.00067216000000}{0.0}
                   {0.495651204385998}{0.869295493713027}
                   {0.703521547759584}{0.408713868967477}
    };
  
    \node[outer,right=8 of 2,label=below:{\tiny $\mathbb{Z}_2 \cong \{E, g _2 \}$}] (7) {
        \massConfig{red}{red}{red}{green}
                   {left}{right}{right}{left}
                   {0}{0}
                   {1.00067216000000}{0}
                   {0.505020955614002}{0.869295493713027}
                   {0.297150612240416}{0.408713868967477}
    };

    \node[outer,right= 8 of 3,label=below:{\tiny $\mathbb{Z}_2 \cong \{E, g _3 \}$}] (8) {
      \massConfig{red}{red}{red}{green}
                 {left}{right}{right}{below}
                 {0}{0}
                 {1.00534612000000}{0}
                 {0.502673060000000}{0.865254047404172}
                 {0.502673060000000}{0.0516266190201723}
    };
  
    \node[outer,right =8 of 4,label=below:{\tiny $D_6$-symmetry}] (9) {
      
        \massConfig{red}{red}{red}{green}
                   {left}{right}{right}{below}
                   {0}{0}
                   {1.00232192000000}{0}
                   {0.501160960000000}{0.868036245489994}
                   {0.501160960000000}{0.289345414510006}
    };
    
    \path (6) edge [<-,out=180,in=0] node[above] {} (5);
    \path (7) edge [<-,out=180,in=0] node[above] {} (5);
    \path (8) edge [<-,out=180,in=0] node[below] {} (5);
    \path (9) edge [<-,out=180,in=0] node[below] {} (5);
    \node[outer,below = of 4] (10) {{\small $ m_4=1 $} };
    \node[outer,right=4 of 10] (11) {{\small$ m_4=m _{\ast\ast} $ }};
    \node[outer,right=3.5 of 11] (12) {{\small $ m_4=0.005 $} };
    \path (10) edge [->,out=0,in=180] node[above] {} (11);%
    \path (11) edge [->,out=0,in=180] node[above] {} (12);%

\end{tikzpicture}
\end{center}
\caption{On the left we show four solutions of the AC equations for $m _1 = m _2 = m _3 = m _4 = 1$. We vary the parameter $ m =m_4 $ and continue these solutions up to $ m  =m_{\ast \ast}  \approx 0.77048695\ldots$, where they coalesce into one solution. This solution branches again into four solutions. These solutions can be continued further. On the right we show the corresponding solutions for $ m _1 = m _2 = m _3 = 1 $ and $ m  = 0.005 $.\label{fig:second_bifurcation} }
\end{figure}
The bifurcations  at $ m _{\ast \ast} $ is not covered by the theory of section \ref{section:bifurcations} because the null-space is two dimensional. This bifurcation was studied in detail in \cite{meyer_bifurcations_1988}
using the Dziobeck equations  and the Liapunov-Schmidt reduction. For the sake of completeness we reproduce those results, but, to differentiate our computations from the ones in \cite{meyer_bifurcations_1988},  we use the Albouy-Chenciner equations instead of the Dziobeck equations. 
  

  Let us denote the Albouy-Chenciner equations for the four-body problem (with normalization $ \lambda ' = - U/(MI) = - 1 $)  as $ f = 0 $ where  $ f = (f _1 , f _2 , f _3 , f _4 , f _5) $, and order the $6 $ variables by introducing the $6 $-vector $ z = (z _1 , z _2 , z _3, z _4 , z _5 , z _6)= (r _{ 12 } , r _{ 13 } , r _{ 14 } , r _{ 23 } , r _{ 24 } , r _{ 34 })$.  
  The equilateral triangle family corresponds to  the solution $ z = a $ where  $ a  = ( \alpha , \alpha ,\frac{  \sqrt{3}}{3} \alpha , \alpha ,\frac{ \sqrt { 3 }}{3} \alpha , \frac{ \sqrt { 3 } }{3} \alpha ) $ with 
   \[
      \alpha = \left( \frac{ 3(\sqrt{3}m + 3)  }{(m + 3)} \right) ^{1/3}.
  \]

    When $m = m _{\ast \ast} $ the $ 4 \times 4 $ submatrix obtained from the Jacobian of the Albouy-Chenciner equations by deleting the last two rows and columns has non-zero determinant. 

  Let  $ L_0 = D f _{ m _{ \ast \ast }} (a) $ be the linearization of $f$ at $ a $. 
  Let $ Q $ be the projection onto the subspace  $  \operatorname{span } ( \{e _5 , e _6  \}) $  of $ \mathbb{R}  ^6 $, where $ e _1 , \ldots , e _6 $ are the elements of the standard basis of $ \mathbb{R}  ^6 $. 
  Let $P$ be the projection onto the subspace $\operatorname{span } (    \{e _5 , e _6  \})$. 
  Let $ z =  Pz + (Id - P) z = u + v $, so that $ u = (0,0,0,0, z _5, z _6) $ and $ v = ( z _1 , z _2 , z _3 ,z _4 , 0 ,0) $. The original equation $ f_m ( z) = 0 $ can now be split into the two equations:
  \[
      \tilde f _m (u,v) =  (Id-Q)f_m (u + v) = 0, \mbox{ and } Q f_m ( u + v) = 0.   
  \] 
  By the implicit function theorem, since the Jacobian  of $  \tilde f _m  $ with respect to $v$ is non-singular, the first of the equations above has the unique solution $ v = v ^\ast_m  ( u) $ for $ m $ near $ m _{\ast \ast}$. This solution can be substituted in the second equation and yields the so-called bifurcation equation 
  \[
  G_m( u) = Q f_m ( u + v ^\ast_m  (u)) = 0.
  \] 
  In our case we use an approximation of $ v ^\ast_m ( u)  $ and $ G_m (u) $ by Taylor expansion. More precisely, let   $ m = m _{\ast \ast} + \epsilon $ , $ z = a + \epsilon b + \epsilon ^2 c + \ldots $, where $ a = ( \alpha , \alpha ,\frac{  \sqrt{3}}{3} \alpha , \alpha ,\frac{ \sqrt { 3 }}{3} \alpha , \frac{ \sqrt { 3 } }{3} \alpha ) $,  $ b = (b _1 , b _2 , b _3 , b _4 , b _5 , b _6) $ , $ c = (c _1 , c _2 ,c _3 ,c _4 ,c _5 , c _6) $, and $ \alpha $ is as above.  We solve    $ (Id-Q)f_m (u + v) = 0$  order by order and we substitute
  into the bifurcation equation $ G _{ m} (u) = 0 $. This allows us to find $ b _1 , b _2 , b _3 , b _4 , c _1 , c _2 , c _3 , c _4 $ as functions of $ b _5, b _6 , c _5 , c _6 $ . In particular we have
  \begin{equation}\label{eqn:bees}  
  \begin{split} 
  b _1 & = \frac{ 1 } { 83 } (81 + 64 \sqrt{3 }) b _6 \\
  b _2 & = \frac{ 1 } { 83 } (81 + 64 \sqrt{3 }) b _5 \\
  b _3 & = - b _6 - b _5 \\
  b _4 & = - \frac{ 1 } { 83 } (81 + 64 \sqrt{3 }) (b _5 + b _6)   
  \end{split}
  \end{equation} 
 we omit the expressions for the $ c _i $s since they are quite long. 
The  equation  $ G _{ m} (u) = 0 $ is identically zero at order $0$ and $1$, while at order 2 becomes:
  \begin{align*} 
      (b_6 + 2b_5)(p_1 b _6  +  p _2)= 0\\
      (b_5+2b_6)(p_1 b _5  + p _2  )= 0
  \end{align*} 
 where
 \begin{align*} 
 p _1 & = 529935346928\\
 p _2 & = 2 ^{ 1/3 } (49+9 \sqrt{3})^{1/3}(207 + 16 \sqrt{3})^{2/3}(362080075 \sqrt{3}-711993501) \\
      & \approx -1.63211356 \times 10 ^{10}
 \end{align*} 
These equations have the following four solutions:
\begin{align*} 
b _5 & = b _6= 0 \\
b _5 & =b _6 = - p _3 \\
b _5 & = -p _3, \quad b _6 = 2 p _3 \\
b _5 & = 2 p _3, \quad b _6 = - p _3 \\
\end{align*} 
where $ p _3 = p _1 / p _2 \approx - 32.46926929$. From this we find  four approximate solutions of the form
{\footnotesize
    \[
    \left(  \alpha + \epsilon b _1 + \ldots  , \alpha + \epsilon b _2 + \ldots ,\frac{ \sqrt{3}}{3}\alpha + \epsilon b _3 + \ldots   , \alpha +  \epsilon b _4 + \ldots  , \frac{ \sqrt{3}}{3}\alpha + \epsilon b _5 + \ldots  , \frac{ \sqrt{3}}{3}\alpha + \epsilon b _6 + \ldots  \right) 
\]
}
where $b _1 , \ldots , b _4 $ can be computed from equations \eqref{eqn:bees}.
These results seem to be compatible  with the results obtained in \cite{meyer_bifurcations_1988} for the Dziobeck equations.
Note that this analysis is local in nature, while our numerical results show that the branches of the bifurcation can be  continued further see figure \ref{fig:second_bifurcation}. The symmetry of the various branches is easy to detect numerically and is indicated in figure  \ref{fig:second_bifurcation}. The symmetry of the solutions can also be inferred theoretically from the symmetry of the equations, see for example the argument in \cite{meyer_bifurcations_1988}.

\section{The Case of Two Pairs of Equal Masses}
\label{sn:2eqmasses}
\subsection{Equivariance}
Recall that the Klein four-group is the group $ \mathbb{Z}_2   \times \mathbb{Z}_2$, the direct product of two copies of the cyclic group of order 2. This group has four elements $   \mathbb{Z}_2   \times \mathbb{Z}_2 = \{ E , h _1 , h _2 , h _3 \} $ with $ h _3 = h _1 h _2 $ and Cayley table 

\[\begin{array}{c|cccc} 
     \circ  &  E   &   h _1  &  h _2  &  h _3  \\ 
    \hline\\  [-2ex]
     E    &  E   & h _1  & h _2  & h _3 \\
     h _1 & h _1 & E     & h _3  & h _2  \\
     h _2 & h _2 & h _3  & E     & h _1  \\
     h _3 & h _3 & h _2  & h _1  & E \\
 \end{array}\] 
where $ E $ is the identity. The proper subgroups of the Klein four-group are $ \{ E \} $ (the trivial group), $ \{ E, h _1 \} $, $ \{ E, h _2 \} $. 
Consider the four body with masses $ m _1 = m _2 = 1 $ and $ m _3 = m _4 = m $, and consider the action $ \Psi $ of the Klein four-group on $ \mathbb{R}  ^8 $ defined by
\

 \begin{align*}
   & \Psi _E =e: (\lambda _0 , \mu , r _{ 12 } , r _{ 13 } , r _{ 14 } , r _{ 23 } , r _{ 24 } , r _{ 34 }) \rightarrow (\lambda _0 , \mu , r _{ 12 } , r _{ 13 } , r _{ 14 } , r _{ 23 } , r _{ 24 } , r _{ 34 }) \\ 
  & \Psi _{h _1 }  =\gamma _1 : (\lambda _0 , \mu , r _{ 12 } , r _{ 13 } , r _{ 14 } , r _{ 23 } , r _{ 24 } , r _{ 34 }) \rightarrow (\lambda _0 , \mu , r _{ 12 } , r _{ 23 } , r _{ 24 } , r _{ 13 } , r _{ 14 } , r _{ 34 }) \\ 
  &  \Psi _{h _2 } = \gamma _2 : (\lambda _0 , \mu , r _{ 12 } , r _{ 13 } , r _{ 14 } , r _{ 23 } , r _{ 24 } , r _{ 34 }) \rightarrow (\lambda _0 , \mu , r _{ 12 } , r _{ 14 } , r _{ 13 } , r _{ 24 } , r _{ 23 } , r _{ 34 }) \\  
  &  \Psi _{h _3 } = \gamma _3 : (\lambda _0 , \mu ,r _{ 12 } , r _{ 13 } , r _{ 14 } , r _{ 23 } , r _{ 24 } , r _{ 34 }) \rightarrow (\lambda _0 , \mu , r _{ 12 } , r _{ 24 } , r _{ 23 } , r _{ 14 } , r _{ 13 } , r _{ 34 }) \\
 \end{align*} 
Then, for each fixed value of $ m=m _4 $  we can think of   
the Dziobeck equations as a map $ F: \mathbb{R}  ^8 \to \mathbb{R}  ^8 $.  A computation shows that this map is  equivariant with respect to the action $ \Psi $ for each value of $ m $.
\subsection{The global picture}
In this subsection we give an overview of the bifurcations we found in the case $ m _1 = m _2 = 1$ and $ m _3 = m _4=m $. The approach taken here is analogous to  the approach taken in subsection \ref{ssn:global_picture_1}, in particular the description presented here is based on  numerical computations.
Without loss of generality, we  restrict our discussion to the case $ 0< m \leq 1 $. 
 We found that the Jacobian of the  Dziobeck (and Albouy-Chenciner) equations vanishes along certain solutions for $ m = \tilde m _\ast \approx 0.99229944 \ldots  $ and  $ m =  \tilde m _{ \ast\ast } \approx 0.99729401 \ldots $. Each of these values of $m$ corresponds to a bifurcation.

 At $ m = \tilde m _{\ast\ast} $  there are two fold bifurcations. In this case, as $ m $ is decreased through $ \tilde m _{\ast\ast} $, four solutions coalesce to two. These solutions are illustrated in Figure \ref{fig:fourth_bifurcation}.

 At $ m = \tilde m _\ast $ there are two supercritical pitchfork bifurcations, so that, when $m $ is decreased through $ \tilde m _\ast $, six solutions coalesce to two. These solutions are illustrated in Figure \ref{fig:fifth_bifurcation}. 

 The number of solutions to the Dziobeck and Albouy-Chenciner equations together with the number of geometrically distinct planar central configurations implied by our numerical computations  is summarized in the following table

 \begin{center}
    \begin{tabular}{| l | l | l | l |}
    \hline
     Value of &  \# of solns & \# of solns & \# of geometrically  \\ 
     $ m_3 = m _4 $ &  of Dziobek eqns&  of AC eqns & different c.c's  \\ \hline\hline
    $ (0,\tilde m _\ast )$ & 11 & 24 & 23\\ \hline
    $ \tilde m_\ast$  & 13 & 26 & 25\\ \hline
    $(\tilde m_\ast, \tilde m _{ \ast \ast }]  $ & 15 & 28 & 27\\ \hline
    $ (\tilde m_{\ast\ast}, 1]  $  & 19 & 32 & 31\\ \hline 
    \end{tabular}
\end{center}
\subsection{Bifurcation at \texorpdfstring{$ m=\tilde m _{ \ast\ast} \approx  0.99729401$}{}}

\begin{figure}[!htbp]
\begin{center}
\begin{tikzpicture}[remember picture,inner/.style={anchor=center},outer/.style={},scale=1, every node/.style={scale=1}]
    \node[outer,label= below:{\tiny $ D_6 $-symmetry }] (1) {
        \massConfig{red}{red}{green}{green}
                   {left}{below}{right}{right}
                   {0.000000000000000}{0.000000000000000}
                   {0.733184690657000}{0.000000000000000}
                   {-0.367154278666138}{0.634869584674111}
                   {-0.367154278666138}{-0.634869584674111}
    };
    \node[outer,below=0.4 of 1,label=below:{\tiny $ \mathbb{Z}_2 \cong \left\{ {E, h _2 } \right\} $}] (2) {
        \massConfig{red}{red}{green}{green}
                   {left}{below}{right}{right}
                   {0.000000000000000}{0.000000000000000}
                   {0.723058256373000}{0.000000000000000}
                   {-0.402471109137504}{0.619745489633413}
                   {-0.402471109137504}{-0.619745489633413}
    };

    \node[outer,below=0.4 of 2,label=below:{\tiny $ \mathbb{Z}_2 \cong \left\{ {E, h _2 } \right\} $}] (3) {
        \massConfig{red}{red}{green}{green}
                   {below}{right}{left}{left}
                   {0.000000000000000}{0.000000000000000}
                   {0.723058256373000}{0.000000000000000}
                   {1.12552936551050}{0.619745489633413}
                   {1.12552936551050}{-0.619745489633413}
    };
    \node[outer,below=0.4 of 3,label=below:{\tiny $ D_6 $-symmetry }] (4) {
        \massConfig{red}{red}{green}{green}
                   {below}{right}{left}{left}
                   {0.000000000000000}{0.000000000000000}
                   {0.733184690657000}{0.000000000000000}
                   {1.10033896932314}{0.634869584674111}
                   {1.10033896932314}{-0.634869584674111}
    };

    \node[outer, above right =2.3of 3,label=below:{\tiny{$ \mathbb{Z}_2 \cong \left\{ {E, h _2 } \right\} $}}] (5) {
          \massConfig{red}{red}{green}{green}
                     {left}{below}{right}{right}
                     {0.000000000000000}{0.000000000000000}
                     {0.728302349968000}{0.000000000000000}
                     {-0.385124631367440}{0.627433128281624}
                     {-0.385124631367440}{-0.627433128281624}
    };

    \node[outer,below right=2.3 of 2,label=below:{\tiny{$ \mathbb{Z}_2 \cong \left\{ {E, h _2 } \right\} $}}] (6) {
        \massConfig{red}{red}{green}{green}
                   {below}{right}{left}{left}
                   {0.000000000000000}{0.000000000000000}
                   {0.728302349969000}{0.000000000000000}
                   {1.11342698133592}{0.627433128280769}
                   {1.11342698133592}{-0.627433128280769}
    };

    \path (1) edge [->,out=0,in=180, pos=0.75] node[above=1.5] {} (5);
    \path (2) edge [->,out=0,in=180] node[below] {} (5);
    \path (3) edge [->,out=0,in=180,pos=0.75] node[above=1] {} (6);
    \path (4) edge [->,out=0,in=180] node[below] {} (6);
    \node[outer,below =0.5 of 4] (7) {{\small$ m_3=m_4=1 $} };
    \node[outer,right=2 of 7] (8) {{\small$ m_3=m_4=\tilde m _{ \ast\ast} $ }};
    \path (7) edge [->,out=0,in=180] node[above] {} (8);%
\end{tikzpicture}
\end{center}
\caption{On the left we show two pairs of solutions for $m _1 = m _2 = m _3 = m _4 = 1$. These solutions are continued by increasing the parameter $ m _3 = m _4=m $. Then, at $ m =\tilde m _{ \ast\ast} \approx  0.99729401...$, each pair of solution  coalesce into one solution with a $ \mathbb{Z}_2 \cong \left\{ {E, h _2 } \right\} $ symmetry. This solution cannot be continued further, since we encounter a fold bifurcation.  \label{fig:fourth_bifurcation}} 
\end{figure}

We now use interval arithmetic to analyze one of the fold bifurcations at $ m = \tilde m _{ \ast\ast} $ (the one with $ m _1 $ in the convex hull formed by the other masses). Let $ \tilde F = [(F _1 , \ldots , F _8, \det(DF)] $  be the vector having as components the Dziobeck equations and the determinant of the Jacobian matrix of $F$. Then we can use  the Krawczyk operator to prove the existence of a (unique) solution $(\mathbf{\tilde x}  _{ \ast \ast}, \tilde m _{ \ast\ast} )$ to the equation $  \tilde F (\mathbf{x} , m ) = 0 $ in a  small box. Let  $ [\mathbf{\tilde x} _{ \ast\ast}]\times [\tilde m_{\ast\ast} ]$ be the box containing the solution  $(\mathbf{\tilde x}  _{\ast\ast} , \tilde m _{\ast\ast} )$. 

Using as initial guess a value obtain using numerical computations we obtain that
\[[\mathbf{\tilde x} _{\ast\ast}]=\begin{bmatrix}
        4.08429981829230230011485100912356858215517?\\
        0.78045312314450202651992?\\
        0.5737849085182166770049?\\
        0.58001687737574791204967?\\ 
        0.58001687737574791204967?\\
        1.0069084529737404291463?\\
        1.0069084529737404291463?\\
        0.9886256052963805814736?\\
       \end{bmatrix}\]
and $ [\tilde m _{ \ast\ast }  ]=  0.997294013195487928197522256274082374264547? $.
Suppose $ A = DF ([\mathbf{\tilde x}  _{ \ast\ast} ],[\tilde  m _{ \ast\ast}]) $. Computing the echelon form of $A$ using Gauss elimination it is possible to show rigorously that the null-space  of $A$ is one dimensional. The eigenvectors of $ A $ and $ A ^T $ corresponding to the zero eigenvalue are  
\[\mathbf{v} = 
\left[\begin{array}{r}
0.? \times 10^{-9} \\
-0.179026448? \\
2.989514215? \\
-0.5496816801? \\
-1.4331568126? \\
-0.5496816801? \\
-1.4331568126? \\
1
\end{array}\right], \quad \mathbf{w} =
\left[\begin{array}{r}
0.? \times 10^{-17} \\
-0.23318293319421040? \\
0.990420115347107375? \\
-0.533270542375855490? \\
-0.533270542375855490? \\
-0.4619301897243832226? \\
-0.4619301897243832226? \\
1
\end{array}\right]\]
respectively.
Moreover we have that 
\begin{align*} 
    & \mathbf{w} ^T F _m ([\mathbf{\tilde x}  _{ \ast \ast  } ],[\tilde m _{ \ast\ast }  ]) = 6.32247017553985546? \\
    & \mathbf{w} ^T [D^2F ([\mathbf{\tilde x}  _{ \ast \ast } ],[\tilde m _{ \ast \ast }  ])(\mathbf{v} ,\mathbf{v} ) ]= -227.08976277782379?
\end{align*}  
and thus, since the interval obtained do not contain zero,  by Theorem \ref{thm:sotomayor}, the bifurcation occurring at   $(\mathbf{\tilde x}  _{ \ast \ast }  , \tilde m _{ \ast\ast} )$ is a fold bifurcation.

\subsection{Bifurcation at \texorpdfstring{$ m=\tilde m _\ast \approx  0.99229944$}{}}

\begin{figure}[!htbp]
\begin{center}
\begin{tikzpicture}[remember picture,inner/.style={anchor=center},outer/.style={},scale=1, every node/.style={scale=1}]
    \node[outer,label= below:{\tiny $ D_6 $-symmetry }] (1) {
        \massConfig{red}{red}{green}{green}
                   {left}{right}{right}{right}
                   {0.000000000000000}{0.000000000000000}
                   {1.27013537434668}{0.000000000000000}
                   {0.635067687173345}{1.09996950042949}
                   {0.635067687173342}{0.366656500143162}
    };
    \node[outer,below=0.4 of 1,label=below:{\tiny $ \mathbb{Z}_2$-symmetry}] (2) {
        \massConfig{red}{red}{green}{green}
                   {left}{right}{right}{right}
                   {0.000000000000000}{0.000000000000000}
                   {1.28504674845946}{0.000000000000000}
                   {0.687678307954183}{1.08556146325093}
                   {0.633349624824445}{0.348522944562675}
    };

    \node[outer,below=0.4 of 2,label=below:{\tiny $ \mathbb{Z}_2$-symmetry}] (3) {
        \massConfig{red}{red}{green}{green}
                   {left}{right}{right}{right}
                   {0.000000000000000}{0.000000000000000}
                   {1.28504674845946}{0.000000000000000}
                   {0.597368440505274}{1.08556146325093}
                   {0.651697123635011}{0.348522944562675}
    };
    \node[outer,right=1.0 of 1, label= below:{\tiny $ \mathbb{Z}_2 \cong \left\{ {E, h_1} \right\} $}] (4) {
        \massConfig{red}{red}{green}{green}
                   {left}{right}{right}{right}
                   {0.000000000000000}{0.000000000000000}
                   {1.28025682000000}{0.000000000000000}
                   {0.640128410000000}{1.09011990515304}
                   {0.640128410000000}{0.353546588648903}
    };
    \node[outer,right=1.0 of 2, label=below:{\tiny no symmetry}] (5) {
        \massConfig{red}{red}{green}{green}
                   {left}{right}{right}{right}
                   {0.000000000000000}{0.000000000000000}
                   {1.28448897000000}{0.000000000000000}
                   {0.669030300188609}{1.08588825561314}
                   {0.636807948923096}{0.348333850998722}
    };

    \node[outer,right=1.0 of 3,label=below:{\tiny no symmetry}] (6) {
        \massConfig{red}{red}{green}{green}
                   {left}{right}{right}{right}
                   {0.000000000000000}{0.000000000000000}
                   {1.28448897000000}{0.000000000000000}
                   {0.615458669811391}{1.08588825561314}
                   {0.647681021076904}{0.348333850998722}
    };

    \node[outer,below=0.4 of 3,label=below:{\tiny $ \mathbb{Z}_2 $-symmetry }] (7) {
        \massConfig{red}{red}{green}{green}
                   {left}{right}{right}{right}
                   {0.000000000000000}{0.000000000000000}
                   {1.28504674845946}{0.000000000000000}
                   {0.651697123635010}{0.348522944562676}
                   {0.597368440505275}{1.08556146325093}
    };

    \node[outer,below=0.4 of 7,label=below:{\tiny $ \mathbb{Z}_2 $-symmetry}] (8) { 
        \massConfig{red}{red}{green}{green}
                   {left}{right}{right}{right}
                   {0.000000000000000}{0.000000000000000}
                   {1.28504674845946}{0.000000000000000}
                   {0.633349624824445}{0.348522944562676}
                   {0.687678307954182}{1.08556146325093}
    };

    \node[outer,below=0.4 of 8,label=below:{\tiny $ D_6 $-symmetry }] (9) {
        \massConfig{red}{red}{green}{green}
                   {left}{right}{right}{right}
                   {0.000000000000000}{0.000000000000000}
                   {1.27013537434669}{0.000000000000000}
                   {0.635067687173343}{0.366656500143160}
                   {0.635067687173340}{1.09996950042949}
    };

    \node[outer,right=1.0 of 7,label=below:{\tiny no symmetry }] (10) {
        \massConfig{red}{red}{green}{green}
                   {left}{right}{right}{right}
                   {0.000000000000000}{0.000000000000000}
                   {1.28448897000000}{0.000000000000000}
                   {0.647681021076904}{0.348333850998722}
                   {0.615458669811391}{1.08588825561314}
    };

    \node[outer,right=1.0 of 8,label=below:{\tiny no symmetry}] (11) { 
        \massConfig{red}{red}{green}{green}
                   {left}{right}{right}{right}
                   {0.000000000000000}{0.000000000000000}
                   {1.28448897000000}{0.000000000000000}
                   {0.636807948923096}{0.348333850998722}
                   {0.669030300188609}{1.08588825561314}
    };

    \node[outer,right=1.0 of 9,label=below:{\tiny $\mathbb{Z}_2 \cong \left\{ {E, h_1} \right\} $}] (12) {
        \massConfig{red}{red}{green}{green}
                   {left}{right}{right}{right}
                   {0.000000000000000}{0.000000000000000}
                   {1.28025682000000}{0.000000000000000}
                   {0.640128410000000}{0.353546588648903}
                   {0.640128410000000}{1.09011990515304}
    };

    \node[outer, right =2.3of 5,label=below:{\tiny{$ \mathbb{Z}_2 \cong \left\{ {E, h_1} \right\} $}}] (13) {
        \massConfig{red}{red}{green}{green}
                   {left}{right}{right}{right}
                   {0.000000000000000}{0.000000000000000}
                   {1.28418635000000}{0.000000000000000}
                   {0.642093184824699}{1.08606573696898}
                   {0.642093169312008}{0.348231208201441}
    };

    \node[outer,right=2.3 of 11,label=below:{\tiny{$ \mathbb{Z}_2 \cong \left\{ {E, h_1} \right\} $}}] (14) {
        \massConfig{red}{red}{green}{green}
                   {left}{right}{right}{right}
                   {0.000000000000000}{0.000000000000000}
                   {1.28418635000000}{0.000000000000000}
                   {0.642093175000000}{0.348231197713522}
                   {0.642093184824699}{1.08606573696898}
    };

    \path (1) edge [->] node {} (4);
    \path (2) edge [->] node {} (5);
    \path (3) edge [->] node {} (6);
    \path (7) edge [->] node {} (10);
    \path (8) edge [->] node {} (11);
    \path (9) edge [->] node {} (12);
    \path (4) edge [->,out=0,in=180, pos=0.75] node[above=1.5] {} (13);
    \path (5) edge [->,out=0,in=180] node[below] {} (13);
    \path (6) edge [->,out=0,in=180,pos=0.75] node[above=1] {} (13);
    \path (10) edge [->,out=0,in=180] node[below] {} (14);
    \path (11) edge [->,out=0,in=180] node[below] {} (14);
    \path (12) edge [->,out=0,in=180] node[below] {} (14);
    \node[outer,below =0.5 of 9] (15) {{\small$ m_3=m_4=1 $} };
    \node[outer,below =0.5 of 12] (16) {{\small$ m_3=m_4=0.995 $} };
    \node[outer,right=2.7 of 16] (17) {{\small$ m_3=m_4=\tilde m _{ \ast} $ }};
    \path (15) edge [->] node {} (16);
    \path (16) edge [->] node {} (17);
\end{tikzpicture}
\end{center}
\caption{On the left we show two groups of three solutions for $m _1 = m _2 = m _3 = m _4 = 1$. We continue these solutions by increasing the parameter $ m _3 = m _4=m $. Then, at $ m=\tilde m _\ast \approx  0.99229944...$, each group of solutions  coalesce into one solution with a $ \mathbb{Z}_2$ symmetry. These solutions can be continued further.  \label{fig:fifth_bifurcation}} 
\end{figure}
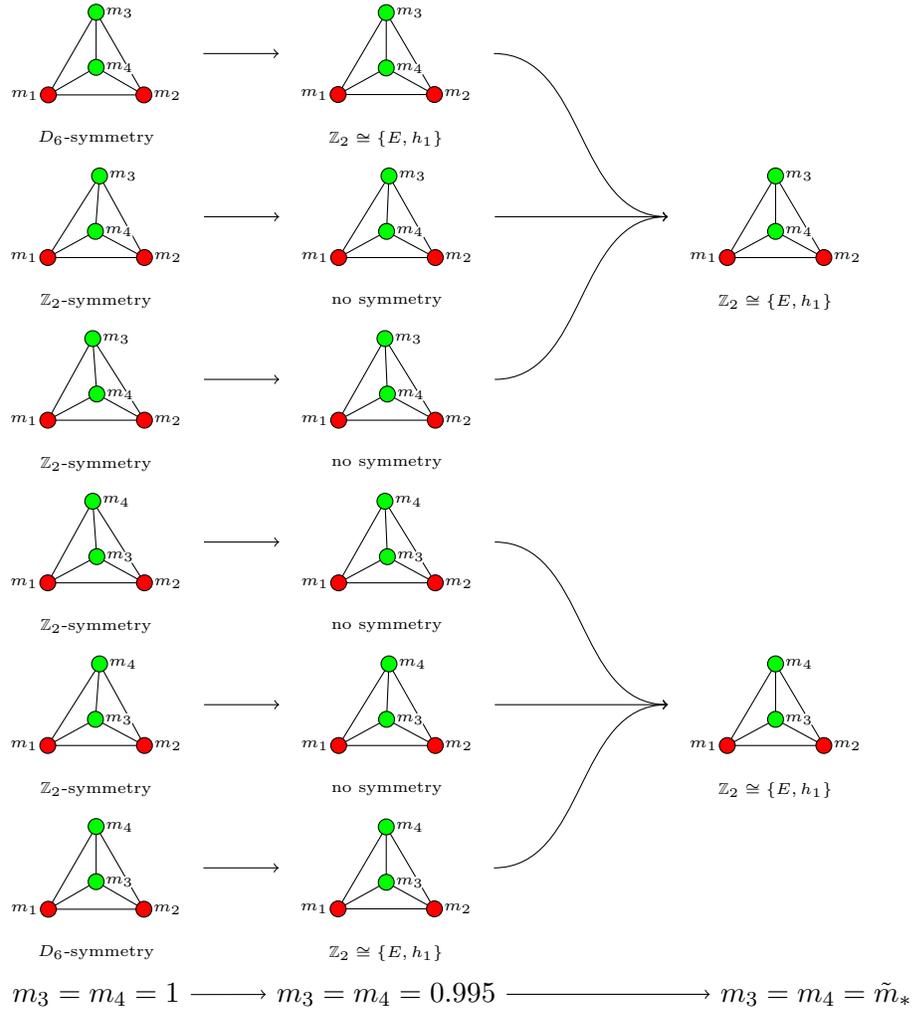

We now use interval arithmetic to analyze one of the pitchfork bifurcations  at $ m = \tilde m _{ \ast } $.
Let $ \tilde F = [(F _1 , \ldots , F _8, \det(DF)] $  be the vector having as components the Dziobeck equations and the determinant of the Jacobian matrix of $F$. To look for this bifurcation we can impose the symmetry $r _{ 13 } = r _{ 23 } $ and $ r _{ 14 } = r _{ 24 } $. 
Let $ G _i $ be equal to $  F _i $ restricted to  $ r _{ 13 } = r _{ 23 }$ and $  r _{ 14 } = r _{ 24 } $ and let $ J $ be t he Jacobian of $ F $ restricted to $ r _{ 13 } = r _{ 23 }$ and $  r _{ 14 } = r _{ 24 } $.  Let $ \tilde G =  [(G _1 , G _2 , G _3,G _6,  \det(J)] $ (since $ G _4 = G _2 $ and $ G _5 = G _3 $ ), then we can use  the Krawczyk operator to prove the existence of a (unique) solution to the equation $  \tilde G (\mathbf{y} , m ) = 0 $, in a  small box. This correspond to proving the existence of a unique symmetric solution $(\mathbf{\tilde x}  _{\ast} , \tilde m _{\ast} )$ for the original equation $\tilde F (\mathbf{x} , m) = 0$ in a small box.  Let  $ [\mathbf{\tilde x} _{\ast}]\times [\tilde m _{\ast} ]$ be the box containing the solution  $(\mathbf{\tilde x}  _{\ast} ,\tilde  m _{\ast} )$. Using as initial guess a value obtain using numerical computations we obtain that
\[[\mathbf{x} _{\ast}]=\begin{bmatrix}
4.0585641815314330056739142?\\
0.7731895057295255894879076?\\
1.0130295438471170352477195?\\
0.57621036528654983921809171?\\
0.99527106304736638582196968?\\
0.57621036528654983921809171?\\
0.99527106304736638582196968?\\
0.58204027784245088387823969?\\
  \end{bmatrix}\]         
and $ [\tilde m _{\ast} ]= 0.9922994477523853474498458?$. 
Suppose $ A = DF ([\mathbf{\tilde x}  _{\ast} ],[\tilde  m _{\ast}]) $. Computing the echelon form of $A$ using Gauss elimination it is possible to show rigourosly that the null-space  of $A$ is one dimensional, since we know that at least one eigenvalue must be zero, but seven of the eight rows of the echelon form are clearly non-zero.   From the echelon form of $A$ we  find that the eigenvectors of $ A $ and $ A ^T $ corresponding to the zero eigenvalue are  
\[\mathbf{v} =
   \left[\begin{array}{r}
    0.? \times 10^{-20} \\
    0.? \times 10^{-21} \\
    0.? \times 10^{-21} \\
    0.3503863414744728128369? \\
    -1.0000000000000000000000? \\
    -0.3503863414744728128369? \\
    1 \\
    0
\end{array}\right]
\quad
    \mathbf{w} =\left[\begin{array}{r}
    0.? \times 10^{-20} \\
    0.? \times 10^{-21} \\
    0.? \times 10^{-21} \\
    1.045364602949539555131? \\
    -1.0000000000000000000000? \\
    -1.0453646029495395551308? \\
    1 \\
    0
\end{array}\right]\]
respectively.
Moreover we have that 
\begin{align*} 
    & \mathbf{w} ^T F _m ([\mathbf{\tilde x}  _{\ast}],\tilde [m _{ \ast}]) =0.? \times 10^{-20} \\
    & \mathbf{w} ^T [DF_m ([\mathbf{\tilde x}  _{ \ast}],[\tilde m _{\ast} ])\mathbf{v}  ]=27.1877227151147526097? \\
    & \mathbf{w} ^T [D^2F ([\mathbf{\tilde x}  _{ \ast}],[\tilde m _{\ast} ])(\mathbf{v} ,\mathbf{v} ) ]=0.? \times 10^{-17} \\
    & \mathbf{w} ^T [D^3F ([\mathbf{\tilde x}  _{\ast}],[\tilde m  _{\ast} ])(\mathbf{v} ,\mathbf{v},\mathbf{v} ) ]= -2639.9736664601674948?\\
\end{align*}  
and thus, using the same argument used in Section \ref{sec:bifurcation_at_0.99184227}, by Theorem \ref{thm:sotomayor} and Lemma \ref{lem:Z2_symmetry}, it follows that  the  bifurcation occurring at   $(\mathbf{x}  _{\ast} , m _{\ast} )$ is a pitchfork bifurcation. Since the last expression above is negative, again by Theorem \ref{thm:sotomayor}, the branches occur for $ m> m _{ \ast} $ and the bifurcation is supercritical. 

\section*{Acknowledgments}
This work was supported by an NSERC Discovery grant. The authors wish to thank James Montaldi for his crucial help in proving Lemma 1. They wish to express their appreciation for helpful email exchanges and discussions with Giampaolo Cicogna, Tsung-Lin Lee, and Winston Sweatman.  

\bibliographystyle{amsplain}
\bibliography{Books,Papers,My_Papers}

\providecommand{\bysame}{\leavevmode\hbox to3em{\hrulefill}\thinspace}
\providecommand{\MR}{\relax\ifhmode\unskip\space\fi MR }
\providecommand{\MRhref}[2]{%
  \href{http://www.ams.org/mathscinet-getitem?mr=#1}{#2}
}
\providecommand{\href}[2]{#2}
\begin{thebibliography}{10}

\bibitem{albouy_symetrie_1995}
Alain Albouy, \emph{Sym{\'e}trie des configurations centrales de quatre corps},
  Comptes Rendus de l'Acad{\'e}mie des Sciences. S{\'e}rie I. Math{\'e}matique
  \textbf{320} (1995), no.~2, 217--220.

\bibitem{albouy_symmetric_1996}
\bysame, \emph{The symmetric central configurations of four equal masses},
  Hamiltonian dynamics and celestial mechanics (Seattle, {WA}, 1995), Contemp.
  Math., vol. 198, Amer. Math. Soc., Providence, {RI}, 1996, pp.~131--135.

\bibitem{albouy_problems_2012}
Alain Albouy, Hildeberto~E. Cabral, and Alan~A. Santos, \emph{Some problems on
  the classical n-body problem}, Celestial Mechanics and Dynamical Astronomy
  \textbf{113} (2012), no.~4, 369--375 (en).

\bibitem{albouy_probleme_1997}
Alain Albouy and Alain Chenciner, \emph{Le probl{\`e}me des n corps et les
  distances mutuelles}, Inventiones Mathematicae \textbf{131} (1997), no.~1,
  151--184.

\bibitem{albouy_symmetry_2008}
Alain Albouy, Yanning Fu, and Shanzhong Sun, \emph{Symmetry of planar four-body
  convex central configurations}, Proceedings of the Royal Society A:
  Mathematical, Physical and Engineering Sciences \textbf{464} (2008),
  no.~2093, 1355--1365.

\bibitem{barros_set_2011}
J.~Barros and E.~Leandro, \emph{The set of degenerate central configurations in
  the planar restricted four-body problem}, {SIAM} Journal on Mathematical
  Analysis \textbf{43} (2011), no.~2, 634--661.

\bibitem{barros_bifurcations_2014}
\bysame, \emph{Bifurcations and enumeration of classes of relative equilibria
  in the planar restricted four-body problem}, {SIAM} Journal on Mathematical
  Analysis \textbf{46} (2014), no.~2, 1185--1203.

\bibitem{bernat_planar_2009}
Josep Bernat, Jaume Llibre, and Ernesto Perez-Chavela, \emph{On the planar
  central configurations of the 4-body problem with three equal masses},
  Dynamics of Continuous, Discrete and Impulsive Systems, Series A:
  Mathematical Analysis \textbf{16} (2009), no.~1, 1--13.

\bibitem{chazy_sur_1918}
Jean Chazy, \emph{Sur certaines trajectoires du probl{\'e}me des n corps},
  Bulletin Astronomique \textbf{35} (1918), 321--389.

\bibitem{corbera_central_2014}
Montserrat Corbera and Jaume Llibre, \emph{Central configurations of the 4-body
  problem with masses and m small}, Applied Mathematics and Computation
  \textbf{246} (2014), 121--147.

\bibitem{hampton_finiteness_2006}
Marshall Hampton and Richard Moeckel, \emph{Finiteness of relative equilibria
  of the four-body problem}, Inventiones mathematicae \textbf{163} (2006),
  no.~2, 289--312 (en).

\bibitem{hampton_relative_2014}
Marshall Hampton, Gareth~E. Roberts, and Manuele Santoprete, \emph{Relative
  equilibria in the four-vortex problem with two pairs of equal vorticities},
  Journal of Nonlinear Science \textbf{24} (2014), no.~1, 39--92 (en).

\bibitem{kuznetsov_elements_2004}
Yuri Kuznetsov, \emph{Elements of applied bifurcation theory}, Springer Science
  \& Business Media, June 2004 (en).

\bibitem{lee_hom4ps-2.0:_2008}
T.~L. Lee, T.~Y. Li, and C.~H. Tsai, \emph{{HOM}4ps-2.0: a software package for
  solving polynomial systems by the polyhedral homotopy continuation method},
  Computing \textbf{83} (2008), no.~2-3, 109--133 (en).

\bibitem{long_four-body_2002}
Yiming Long and Shanzhong Sun, \emph{Four-body central configurations with some
  equal masses}, Archive for Rational Mechanics and Analysis \textbf{162}
  (2002), no.~1, 25--44.

\bibitem{meyer_bifurcations_1988}
Kenneth~R. Meyer and Dieter~S. Schmidt, \emph{Bifurcations of relative
  equilibria in the 4- and 5-body problem}, Ergodic Theory and Dynamical
  Systems \textbf{8} (1988), no.~Volume 8*, 215--225.

\bibitem{moulton_straight_1910}
F.~R. Moulton, \emph{The straight line solutions of the problem of n bodies},
  The Annals of Mathematics \textbf{12} (1910), no.~1, 1--17.

\bibitem{neumaier_interval_2008}
A.~Neumaier, \emph{Interval methods for systems of equations}, Cambridge
  University Press, December 2008 (en).

\bibitem{palmore_relative_1973}
Julian~Ivanhoe Palmore, \emph{Relative equilibria of the n-body problem.},
  Ph.D. thesis, University Of California, Berkley, 1973.

\bibitem{perez-chavela_convex_2007}
Ernesto Perez-Chavela and Manuele Santoprete, \emph{Convex four-body central
  configurations with some equal masses}, Archive for Rational Mechanics and
  Analysis \textbf{185} (2007), no.~3, 481--494.

\bibitem{simo_relative_1978}
C.~Simo, \emph{Relative equilibrium solutions in the four body problem},
  Celestial Mechanics and Dynamical Astronomy \textbf{18} (1978), no.~2,
  165--184.

\bibitem{smale_mathematical_1998}
Steve Smale, \emph{Mathematical problems for the next century}, The
  Mathematical Intelligencer \textbf{20} (1998), no.~2, 7--15 (en).

\bibitem{sotomayor_generic_1973}
Jorge Sotomayor, \emph{Generic bifurcations of dynamical systems}, Dynamical
  systems (Proc. Sympos., Univ. Bahia, Salvador, 1971), Academic Press, New
  York, 1973, pp.~561--582. \MR{0339280}

\bibitem{stein_sage_2014}
W.~A. Stein and {others}, \emph{Sage mathematics software (version 6.1.1)}, The
  Sage Development Team, 2014, http://www.sagemath.org.

\bibitem{wintner_analytical_1964}
Aurel Wintner, \emph{The analytical foundations of celestial mechanics},
  Princeton University Press, 1964.

\bibitem{xia_central_1991}
Zhihong Xia, \emph{Central configurations with many small masses}, Journal of
  Differential Equations \textbf{91} (1991), no.~1, 168--179.

\end{thebibliography}
\end{document}